\documentclass[a4paper,11pt]{article}
\pdfoutput=1
\usepackage{jheppub}
\usepackage{graphicx}
\usepackage{subcaption}
\usepackage{dsfont}
\usepackage[dvipsnames]{xcolor}
\usepackage{etoolbox}
\usepackage[T1]{fontenc}
\usepackage[latin9]{inputenc}
\usepackage{amsmath}
\usepackage{amsthm}
\makeatletter
\patchcmd{\maketitle}{\@fpheader}{\\}{}{}

\theoremstyle{remark}
\newtheorem{claim}{\protect\claimname}
\newtheorem{cor}{\protect\corollaryname}
\newtheorem{lem}{\protect\lemmaname}
\newtheorem{defn}{Definition}[section]
\makeatother

\providecommand{\claimname}{Claim}
\providecommand{\corollaryname}{Corollary}
\providecommand{\lemmaname}{Lemma}
\title{Radial Cutoffs and Holographic Entanglement}

\author[]{Brianna Grado-White,}
\author[]{Donald Marolf,}
\author[]{and Sean J. Weinberg}

\affiliation[]{Department of Physics, University of California, Santa Barbara, CA 93106, USA}

\emailAdd{brianna@physics.ucsb.edu}
\emailAdd{marolf@physics.ucsb.edu}
\emailAdd{sjasonw@physics.ucsb.edu}

\abstract{Tensor networks, $T\bar{T}$, and broader notions of a holographic principle all motivate the idea that some notion of gravitational holography should persist in the presence of a radial cutoff. But in the absence of time-reflection symmetry, the areas of Hubeny-Rangamani-Takayanagi surfaces anchored to the radial cutoff generally violate strong subadditivity, even when the associated boundary regions are spacelike separated as defined by both bulk and boundary notions of causality.  We thus propose an alternate definition of cutoff-holographic entropy using a restricted maximin prescription anchored to a codimension 2 cutoff surface.  For bulk solutions that respect the null energy condition, we show that the resulting areas satisfy SSA, entanglement wedge nesting, and monogamy of mutual information in parallel with cutoff free results in AdS. These results hold even when the cutoff surface fails to be convex.}

\begin{document}

\global\long\def\Mm{\mathrm{Mm}_\gamma}%
\global\long\def\A{\text{Area}}%

\maketitle
\section{Introduction}

There is a great deal of interest in generalizing the AdS/CFT correspondence so as to rely less on the presence of an asymptotically AdS boundary. An ultimate goal would be understand a notion of gravitational duality relevant to cosmology, and in particular to our own apparently-inflating spacetime.

A possible first step toward this goal is start with a standard asymptotically-AdS holographic set up, and then to remove the AdS boundary by introducing a finite radial cutoff.  This was idea behind the work of  \cite{McGough:2016lol} and its generalizations (e.g. \cite{Taylor:2018xcy,Hartman:2018tkw}; see also \cite{Guica:2019nzm}), which posited that the introduction of such cutoffs was related to irrelevant deformations of the dual CFT.  Such radial cutoffs are naturally taken to define codimension 1 boundaries at finite distance from the bulk, though we will emphasize the study of codimension 2 boundaries in sections \ref{sec:prelim}-\ref{sec:disc}.

More generally,  the idea that some notion of holography should persist in the presence of a radial cutoff is strongly motivated by tensor network models; see e.g.\cite{Swingle:2009bg, Qi:2013caa, Evenbly2011,MolinaVilaplana:2011xt, Swingle:2012wq, Matsueda:2012xm,Pastawski:2015qua, Hayden:2016cfa}. In any tensor network, an arbitrary cut through the interior (perhaps defined by a cutoff surface) will define a state living on that cut; see figure \ref{fig:tn}. Furthermore, in many cases where the original tensor network defines an isometric embedding of a bulk Hilbert space into a boundary dual, the same will be true of the cutoff network.  Such ideas are closely related to the surface/state correspondence suggested in \cite{Miyaji:2015yva}, the entanglement of purification conjecture \cite{Takayanagi:2017knl, Nguyen:2017yqw}, and the construction of tensor networks on sub-AdS scales described in \cite{Bao:2018pvs,Bao:2019fpq}.  See also \cite{Krishnan:2019ygy,Krishnan:2020oun}.

\begin{figure}[t]
        \begin{center}
        \begin{subfigure}{.49\linewidth}
        \centering
                \includegraphics[width=0.75\textwidth]{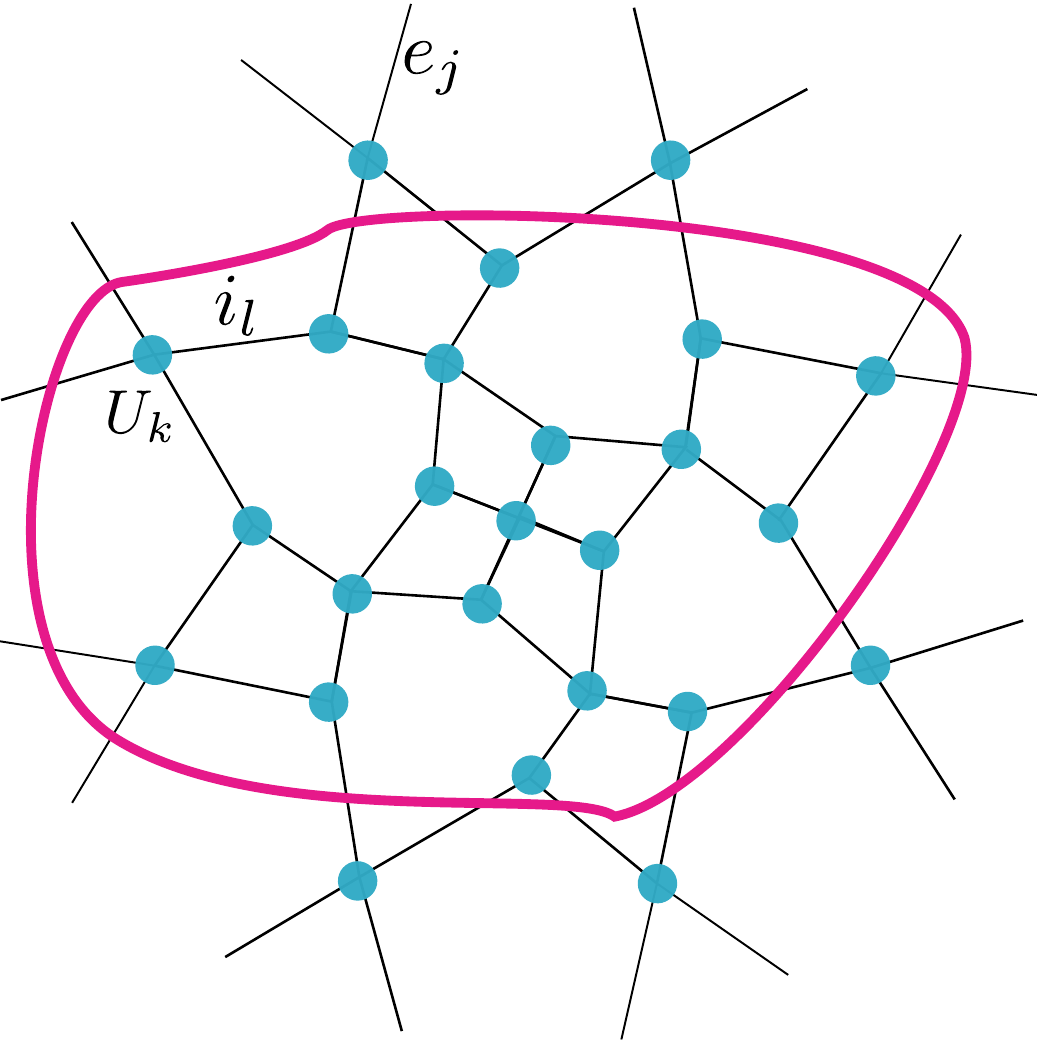}\quad
                \end{subfigure}
		\begin{subfigure}{0.49\linewidth}
        \centering
                \includegraphics[width=0.75\textwidth]{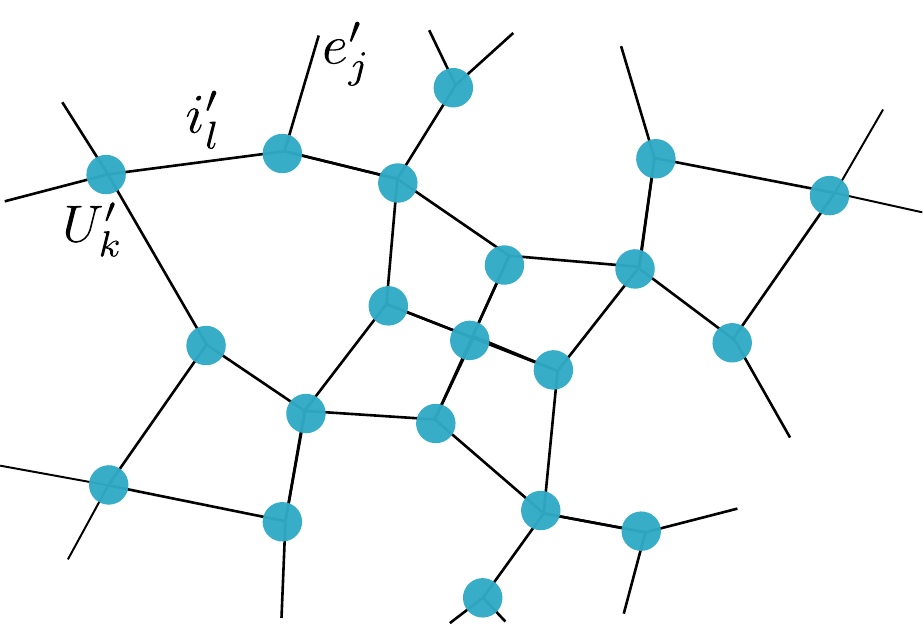}
                \end{subfigure}
        \end{center}
        \caption{A tensor network is a graphical depiction of a quantum state in a multipartite Hilbert space.  It also represents a circuit that prepares the state, consisting of tensors $U_k$ contracted over internal indices $i_l$. Left: A tensor network representing a state $|\psi\rangle \in \mathcal{H}_{e_1}\otimes... \otimes\mathcal{H}_{e_n}$
        and a magenta curve along which this network will be cut off. We use $e_j$ to denote the links at the boundary of the network, though only one is labelled explicitly in the figure. Right: Cutting off the tensor network at left defines a new tensor network which constructs a different state  $|\psi'\rangle \in \mathcal{H}_{e'_1}\otimes...\otimes \mathcal{H}_{e'_{m}}$ in a different Hilbert space associated with the new boundary links $e'_j$ (again with only one labelled explicitly in the figure).}
        \label{fig:tn}
\end{figure}

However, there remains the question of how to access the supposed quantum state on the cutoff surface using bulk techniques, and in particular how to compute quantities related to quantum information.  In the time symmetric context, it would appear that the Ryu-Takayanagi (RT) formula \cite{Ryu:2006ef,Ryu:2006bv} continues to give consistent results for entanglement entropy in the presence of such cutoffs. In particular, in addition to the explicit checks performed in \cite{Donnelly:2018bef}, it is easy to see that the original proof of  strong subadditivity (SSA)  \cite{Headrick:2007km} for contexts without cutoffs in fact continues to hold when one is present.  Indeed, the argument of \cite{Headrick:2007km} relies only on the fact that RT surfaces are minimal subsurfaces on a common time slice, without regard to the nature of the geometry on that slice.

In dynamical contexts, however Hubeny-Rangamani-Takayanagi (HRT) surfaces anchored to a finite boundary may fail to be minimal subsurfaces of a common achronal surface, even when their boundary anchor sets are spacelike separated with respect to the bulk.  In this context,  SSA may simply fail to hold.

We discuss several examples of such SSA violations in sections \ref{sec:counter-flat} and \ref{sec:countercutoff} below.  Other cases violating HRT-SSA with a radial cutoff  were recently studied in \cite{Lewkowycz:2019xse}, where it was found that entanglement wedge nesting failed as well. See also \cite{Geng:2019ruz} for discussions of SSA failures in de Sitter space. The examples of \cite{Lewkowycz:2019xse} were associated with a discrepancy between bulk and boundary causality, as regions of the cutoff surface were spacelike separated with respect to the induced metric on the cutoff surface but were nevertheless causally related through the bulk. One might thus think that the SSA issue could be resolved by restricting to boundary regions that are also spacelike separated through the bulk.
However, the examples of \ref{sec:counter-flat} below will show explicitly that SSA violations can occur even in contexts that respect bulk causality in this way.

Instead, the key feature of all violations turns out to be the failure of the relevant extremal surfaces to be contained in a single domain of dependence within the cutoff bulk. In the terminology of \cite{Miyaji:2015yva}, this is called a failure of convexity. This means that entanglement wedge nesting also fails again. But with regard to SSA, the domain of dependence issue in particular implies that there cannot be a bulk achronal slice containing both the relevant slice of the cutoff boundary and all of the extremal surfaces. As a result, and consistent with the explicit violation in our examples, one cannot readily use the same proof strategy of  \cite{Headrick:2007km} as in the RT case.  In particular, the HRT-SSA proof given for the cutoff free case in \cite{Wall:2012uf} does not apply.

Improving the situation appears to require a new proposal for holographic entropy in the presence of a radial cutoff.
The above comparison with \cite{Wall:2012uf} immediately suggests that we consider instead a maximin based construction,  which in simple cutoff free cases provides an alternate definition of holographic entanglement that turns out to be equivalent to HRT. The maximin approach can also be used to establish SSA in the the case of a convex cutoff surface \cite{Sanches:2016sxy}; see also \cite{Nomura:2018kji,Murdia:2020iac}. Now, as shown by the example in section \ref{sec:counter-flat}, the original maximin construction of \cite{Wall:2012uf} does not suffice.
But realizing that we do not currently understand what notion of causality might govern the propagation of information and excitations in a (likely nonlocal) dual description on the cutoff surface, we will take one further step and consider instead the {\it restricted} maximin procedure of \cite{Marolf:2019bgj} associated with a codimension 2 cutoff surface $\gamma$ rather than a codimension 1 radial cutoff. Restricted maximin surfaces are confined by construction to the domain of dependence of an achronal surface that ends on $\gamma$, so any extremal surface outside this domain must differ from the associated restricted maximin surface. We will show below that the areas of our restricted maximin surfaces {\it do} satisfy SSA, suggesting that these surfaces give a better definition of holographic entanglement in settings with a radial cutoff.  Indeed, we will see that in time-symmetric settings this prescription reproduces the successful RT prescription much better than does naive application of HRT.   Entanglement wedge nesting and monogamy of mutual information \cite{Hayden:2011ag} will follow as well.

\begin{figure}[t]
        \begin{center}
        \begin{subfigure}{.49\linewidth}
        \centering
                \includegraphics[width=0.91\textwidth]{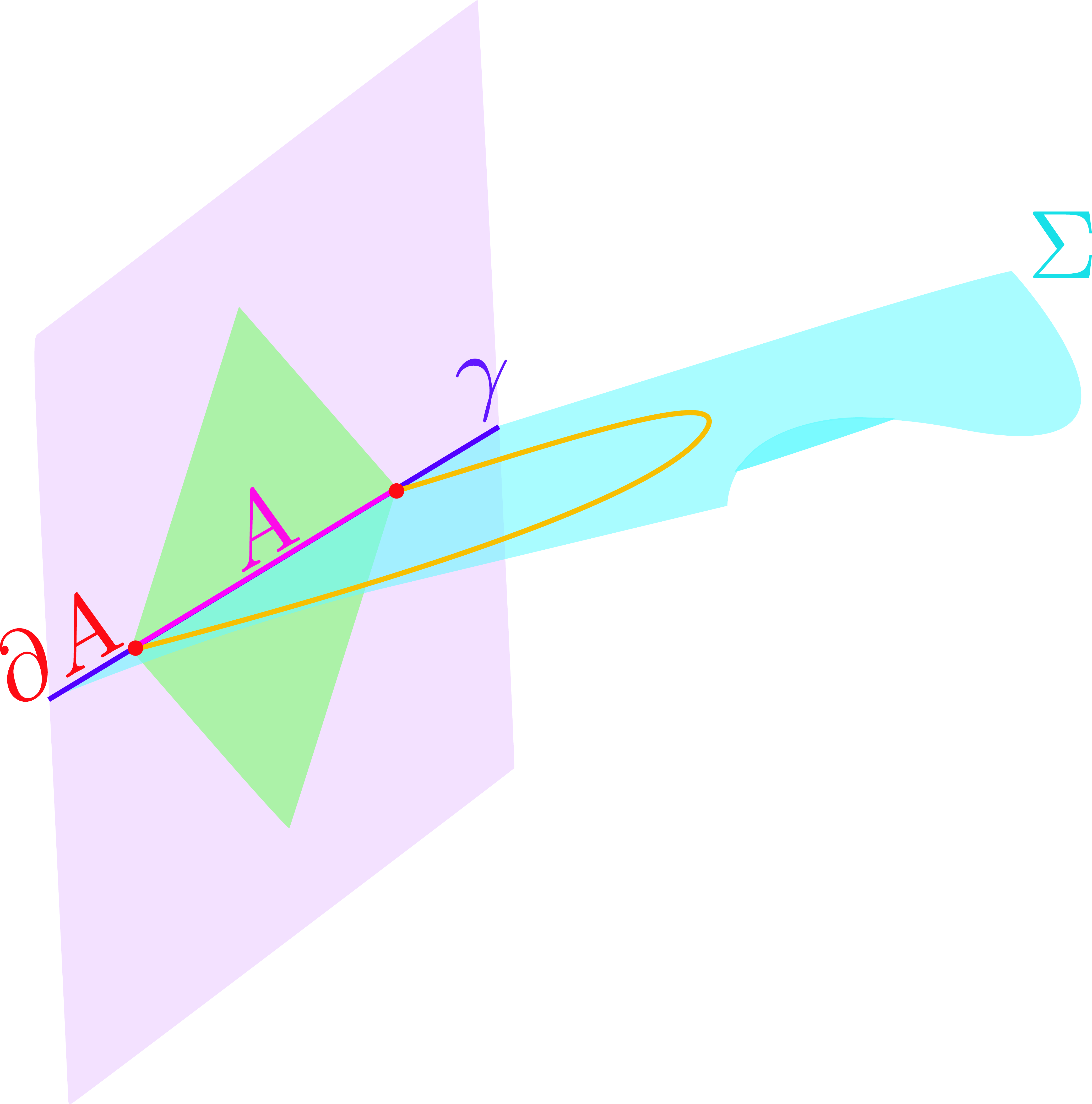}\quad
                \end{subfigure}
		\begin{subfigure}{0.49\linewidth}
        \centering
                \includegraphics[width=0.9\textwidth]{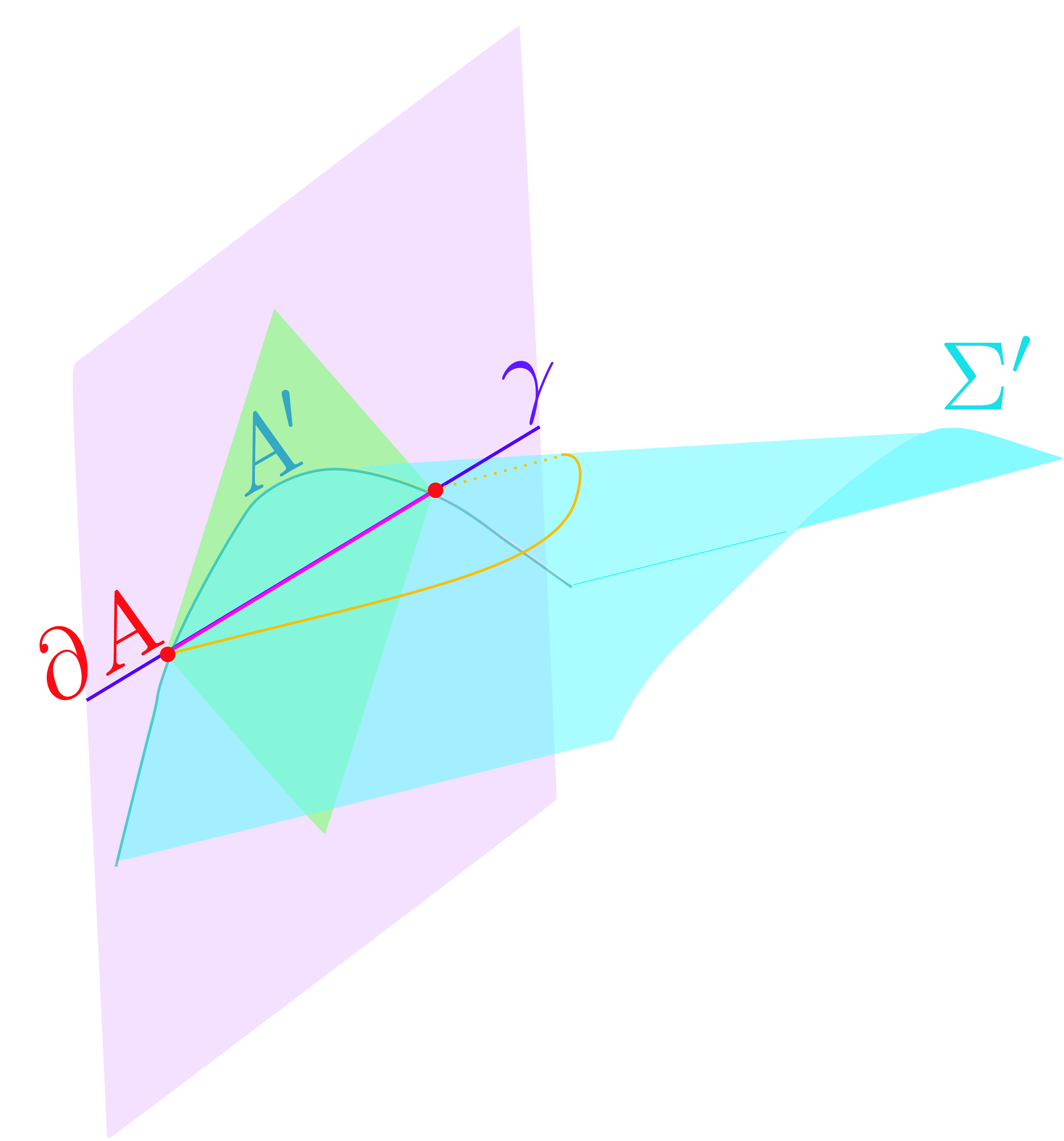}
                \end{subfigure}
        \end{center}
        \caption{Left: A restricted maximin surface is defined by a maximin construction on Cauchy slices which are constrained to contain a codimension $2$ surface $\gamma$. Right: For an (unrestricted) maximin surface, $\Sigma'$ need not contain $\gamma$ or $A$, but may instead meet $\gamma$ only on the subset $\partial A$.}
        \label{fig:restr_maxi}
\end{figure}

In fact, we will work below with a very general notion of ``cutoff.''  We consider any smooth codimension 2 surface in any spacetime satisfying basic positive energy properties (to be detailed below).  In particular, we make no assumption that our spacetime with boundary be constructed from an asymptotically AdS spacetime, or that its dynamics involve a negative cosmological constant.

Before proceeding, we pause to give a very brief summary of the restricted maximin approach from \cite{Marolf:2019bgj}.
The construction begins by choosing a slice $\gamma$ of the cutoff boundary, so that $\gamma$ is a codimension 2 surface with respect to the bulk.  It then constructs a bulk entangling surface for $A\subset \gamma$ using a two step procedure that first minimizes the area of surfaces homologous to $A$ on each bulk achronal surface with boundary $\gamma$.  We call these achronal surfaces `Cauchy surfaces' as we are interested only in the corresponding domain of dependence.   The second step then maximizes the above minima over all such Cauchy surfaces.  We use $\Mm(A)$ to denote the resulting restricted maximin surface.  Here the term `restricted' refers to the fact that in the second step, the boundary of the Cauchy surfaces is held fixed, whereas the original (unrestricted) maximin procedure considered all bulk Cauchy surfaces that include the much smaller set $\partial A$.  In the contexts originally studied in \cite{Wall:2012uf}, unrestricted and restricted maximin surfaces turn out to coincide \cite{Marolf:2019bgj}.


Interestingly, the recent works \cite{Penington:2019npb,Almheiri:2019psf,Almheiri:2019hni,Almheiri:2019yqk,Rozali:2019day,Chen:2019uhq,Bousso:2019ykv,Almheiri:2019psy,Almheiri:2019qdq,Penington:2019kki} on recovering the Page curve from quantum extremal surfaces suggest a rather different motivation for fixing a codimension 2 surface $\gamma$.  In those works it was noted that codimension 2 boundaries are natural when the bulk is not an isolated system.  Time evolution then generally changes the entropies being studied, so one must fix a slice of the codimension 1 boundary (and thus effectively a codimension 2 boundary) to obtain a well-defined answer.  Fixing a codimension 2 boundary is thus natural if the bulk inside the supposed cutoff surface continues to interact with the exterior; i.e., it is natural if one regards the codimension 1 ``cutoff'' as merely a slice through a larger bulk system as opposed to a true cutoff on the dynamics.

The remainder of the paper is organized as follows. We begin in section \ref{sec:prelim} by reviewing known violations of strong subadditivity for HRT areas in the presence of a radial cutoffs.  This section also states various definitions and reviews certain useful results from prior works.  The proofs of restricted maximin strong subadditivity, entanglement wedge nesting, and monogamy of mutual information for an arbitrary codimension 2 boundary $\gamma$ are given in section \ref{sec:results}. Section \ref{sec:disc} concludes with a brief discussion.
An appendix contains some additional results (not required for the main argument) showing that up to sets of measure zero our restricted maximin surfaces either coincide with $\gamma$ or are spacelike separated from $\gamma$; i.e., null separations from $\gamma$ are rare.

\section{Preliminaries}
\label{sec:prelim}
We begin by reviewing known violations of strong subadditivity (SSA) for HRT areas in the presence of a finite radius cutoff. In each case, we will see that SSA is nevertheless satisfied by restricted maximin areas. We then establish notation for the remainder of the paper, formally define our restricted maximin surfaces, and review useful results from \cite{Wall:2012uf} regarding null congruences that touch at a point.

\subsection{Violations of HRT Strong Subadditivity}
\label{sec:counter}

We now review several examples of strong subadditivity violations for HRT areas from \cite{Sanches:2016sxy} and \cite{Hubeny}.  In each example, we will see that the restricted maximin areas in fact satisfy SSA.

\subsubsection{Violations of Subadditivity Deep in the Bulk}
\label{sec:counter-flat}
Our first class of examples was described in \cite{Sanches:2016sxy}. In these situations, the HRT areas violate not only strong subadditivity but also the weaker subadditivity inequality $S_A + S_B \geq S_{AB}$. For simplicity, we consider the case of $d+1=3$ bulk dimensions so that the codimension 2 cutoff surface $\gamma$ is a curve. Specifically, we consider a case where $\gamma$ contains two null geodesic segments $A$ and $B$, whose past ends coincide at some point $p$, and where $A$ and $B$ are small enough that we may approximate the spacetime near them as flat Minkowski space; see figure \ref{fig:plane}. In particular, $A$ and $B$ will then lie in a common timelike plane.  In the approximation that the spacetime is flat, the extremal surface $x(A)$ associated with subregion $A$ coincides with $A$. Similarly, $x(B)$ then coincides with $B$. However, the extremal surface $x(AB)$ will be a spacelike curve  with $S(AB)=|x(AB)|>0$, where $|a|$ denotes the proper length of the surface $a$. But $A$ and $B$ are null, so $S(A) = |x(A)|$ and $S(B) = |x(B)|$ both vanish,  violating subadditivity.

Suppose that we instead wish to use a standard maximin construction as in \cite{Wall:2012uf}.  This would require the specification of a codimension 1 cutoff surface.  If we take this cutoff to include the entire timelike plane containing $A$ and $B$, then the maximin surfaces will again be given by the extremal surfaces $x(A)=A$, $x(B)=B$, and $x(AB)$ given above.  So the same violation would remain.

In contrast, the restricted maximin surfaces $\Mm(A)$, $\Mm(B)$, and $\Mm(AB)$ of $A,B,$ and $AB$ are $\Mm(A)=A$, $\Mm(B)=B,$ and $\Mm(AB)=AB$. This is clear from the fact that their areas all vanish, so they must be minimal on each allowed Cauchy surface. Maximizing zero over all Cauchy surfaces gives zero, so they are  restricted maximin surfaces as claimed. Note that their vanishing areas satisfy both SA and SSA.

\begin{figure}
    \centering
    \includegraphics[width = 0.5\linewidth]{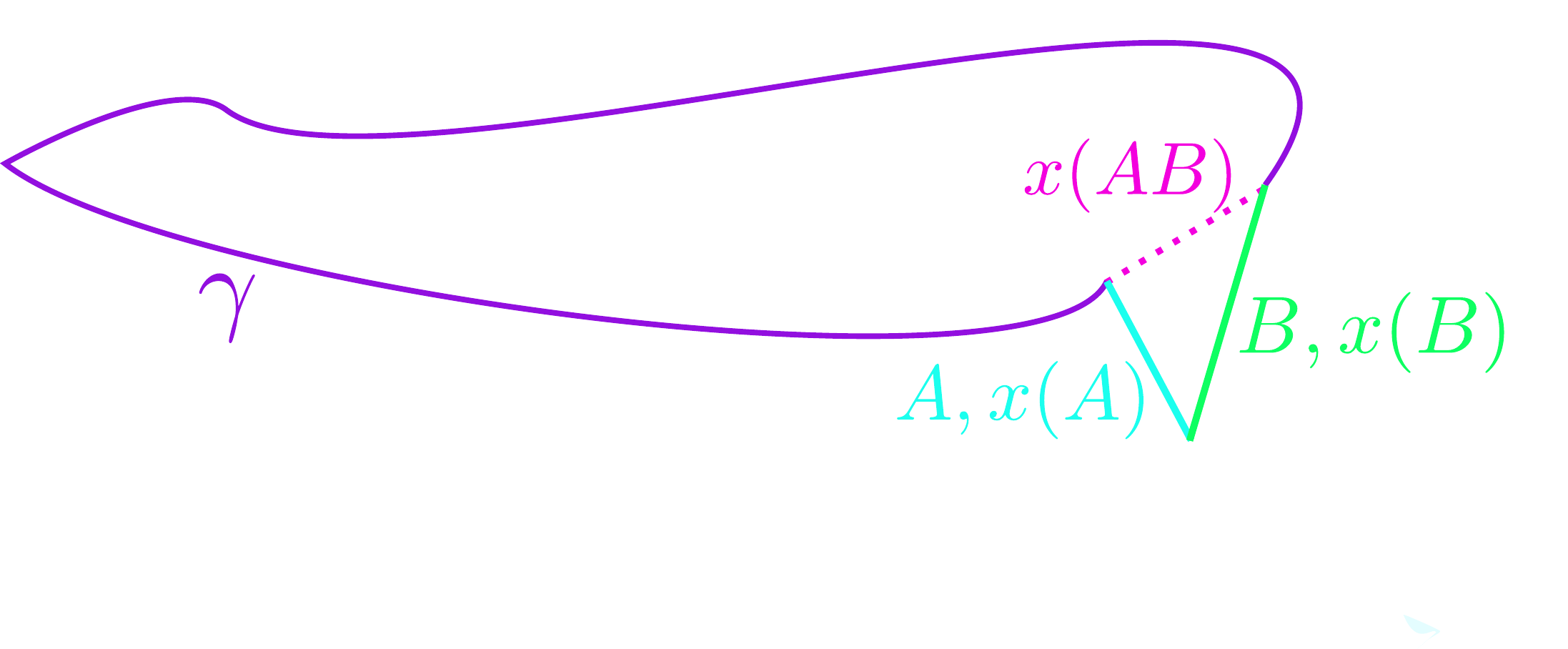}
    \caption{An example violating HRT subadditivity with a cutoff in an approximately flat (2+1) dimensional bulk. On a codimension 2 cutoff $\gamma$ (solid curve), we choose two subregions, $A$ (blue) and $B$ (green).  The cutoff $\gamma$ is constructed so that $A$ and $B$ are null segments intersecting in a timelike plane. The extremal surfaces $x(A)$ and $x(B)$ associated with $A$ and $B$ respectively coincide with $A$ and $B$.  However, the extremal surface $x(AB)$ (dotted pink) associated with $AB$ is a  spacelike curve whose non-zero area violates subadditivity.}
    \label{fig:plane}
\end{figure}

\subsubsection{Violations of Strong Subadditivity Near the Boundary}
\label{sec:countercutoff}
A second violation of HRT-SSA arises when one imposes a simple radial cutoff on empty global AdS$_3$ \cite{Hubeny}. Here we think of the cutoff as defined by a codimension 1 cylinder near the boundary (see figure \ref{fig:radialcounter1}). The construction of the example proceeds in stages. One first finds a situation saturating strong subadditivity. This may be done by starting with a null plane in the bulk, and considering two spacelike boundary intervals $A_0$ and $C_0$ formed by the intersection of this null plane and the cutoff cylinder. One also chooses another spacelike segment $B_0$ that connects $A_0$ and $C_0$ as in figure \ref{fig:radialcounter1}.  In the limit where the cutoff surface becomes the original AdS boundary, $A$ and $C$ become null and the setup resembles both that of \cite{Casini:2006es} and \cite{Lewkowycz:2019xse}.

\begin{figure}[t]
        \begin{center}
        \begin{subfigure}{.49\linewidth}
        \centering
                \includegraphics[width=0.75\textwidth]{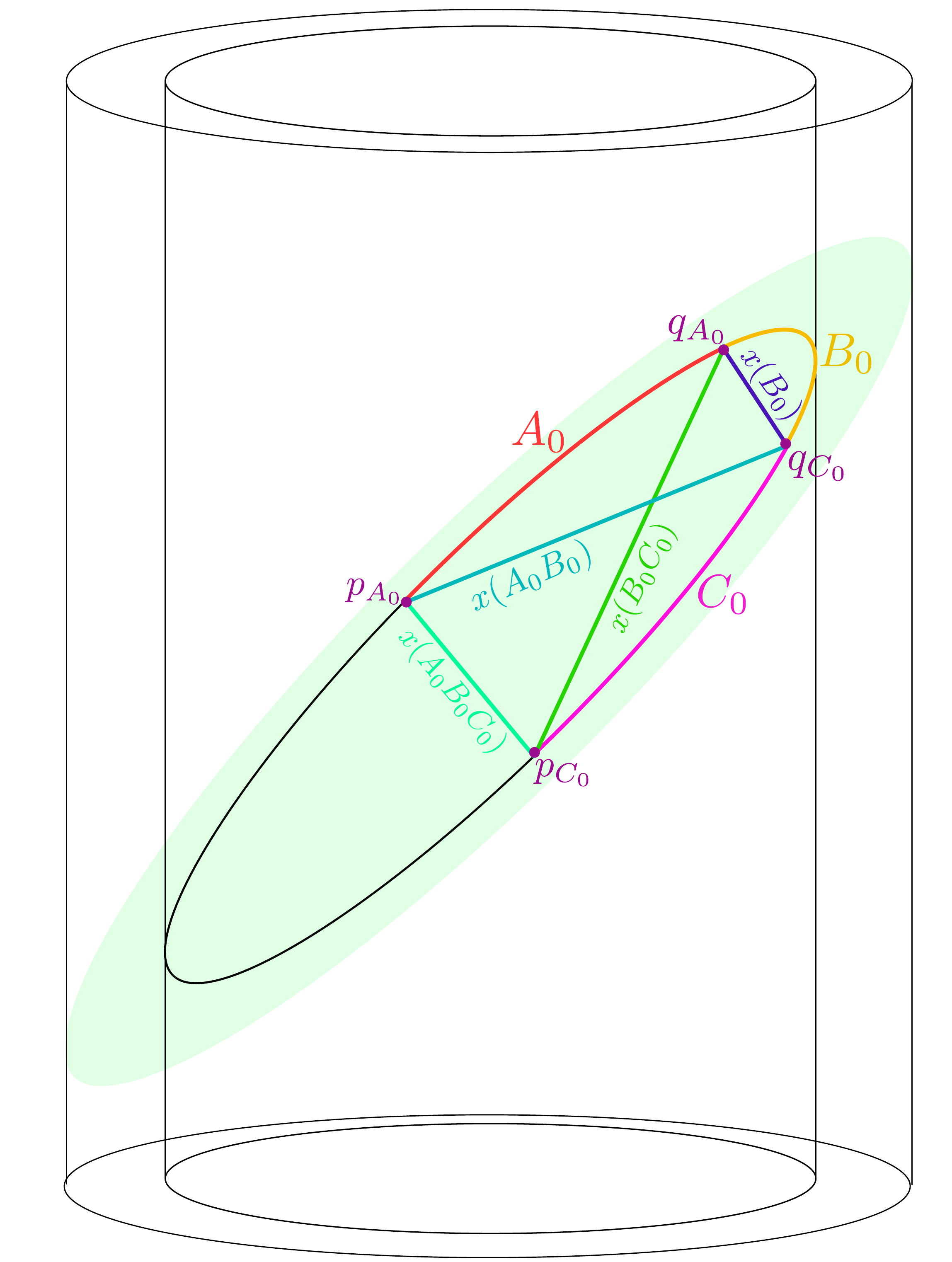}\quad
                \end{subfigure}
		\begin{subfigure}{0.49\linewidth}
        \centering
                \includegraphics[width=0.75\textwidth]{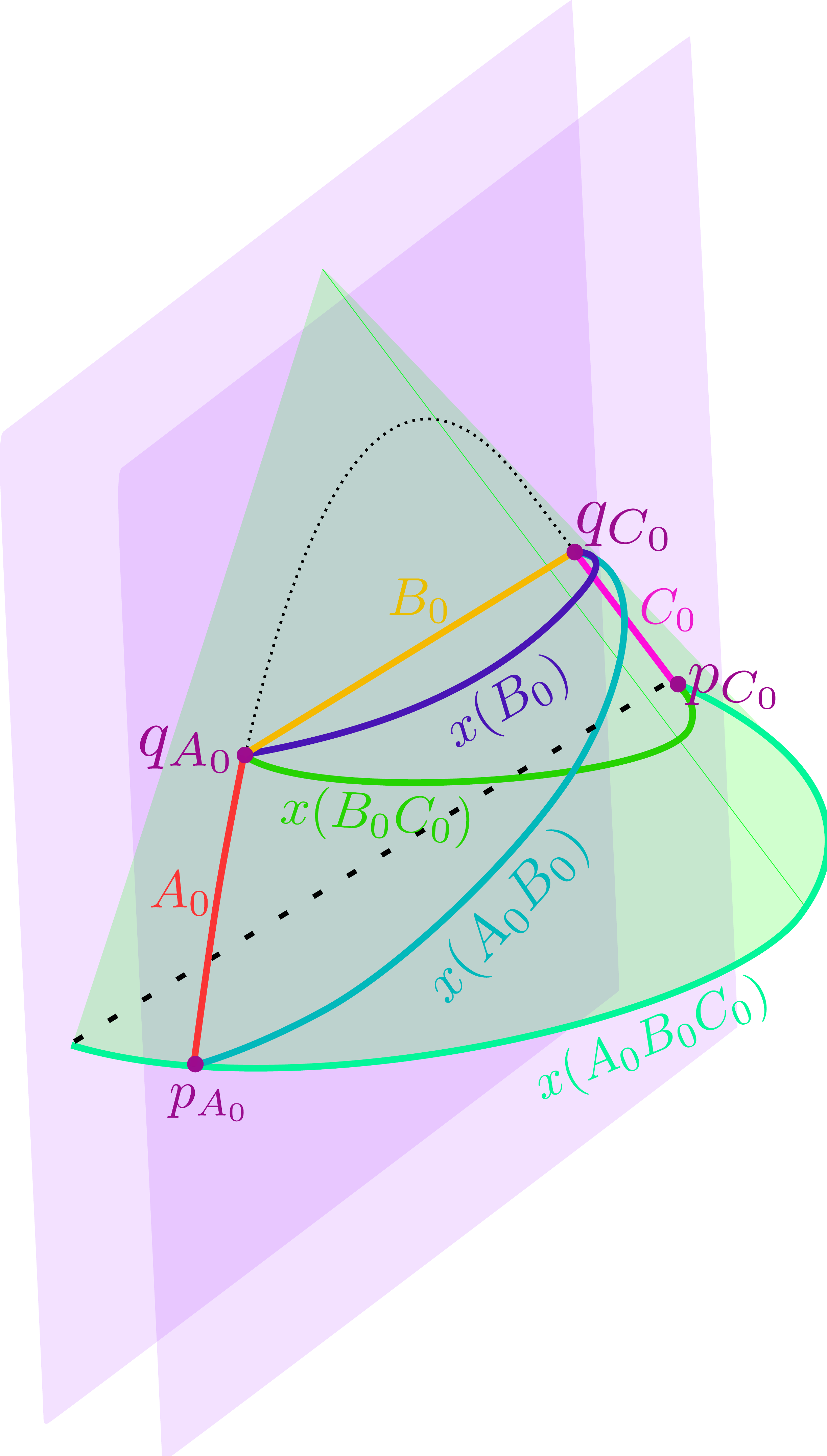}
                \end{subfigure}
        \end{center}
        \caption{An example where HRT areas saturate strong subadditivity on a cutoff surface. Left: A null plane (green) in vacuum AdS$_3$ with a cylindrical cutoff, marked intervals $A_0$, $B_0$, $C_0$, and the corresponding extremal surfaces $x(A_0),x(B_0),x(C_0)$. Right: A rough sketch of the same setting in Poincar\'e coordinates centered on $B_0$. The sketch becomes an exact representation of the conformal structure in the limit where the spacelike intervals $A_0, B_0$, and $C_0$ are small in the sense that they are contained in a small neighborhood near the top of the ellipse in the left panel.}
        \label{fig:radialcounter1}
\end{figure}

Since the bulk spacetime is just empty AdS$_3$, the HRT surfaces are known exactly.  In particular, the familiar statement that entanglement wedges and causal wedges coincide in empty AdS$_3$ means that the HRT surfaces associated with any combination of the regions $A_0,B_0,C_0$ must lie in the above null plane for any value of the cutoff.
And since the expansion of the null congruence generating this null plane vanishes, curves in the plane can be deformed along the plane without changing their length so long as the endpoints of the curves are held fixed.  One thus finds
\begin{equation}
\begin{aligned}
|x(A_0B_0C_0 )| + |x(B_0)|& = |x(A_0)| + 2|x(B_0)|+ |x(C_0)|\\
&= |x(A_0B_0)| + |x(B_0C_0)|\\
\end{aligned}
\end{equation}
so that strong subadditivity is saturated.

The key point is then that this example can be perturbed by translating the future endpoints $q_{A_0}$ and $q_{C_0}$ of $A_0$ and $C_0$ in a future timelike direction to define new future endpoints $q_{A}$ and $q_{C}$ (still spacelike separated from $p_{A_0}$, $p_{C_0}$) and new intervals $A,B$ and $C$. Since the translation is an isometry, one finds $|x(B)| = |x(B_0)|$.  Further, since the past endpoints $p_{A_0}$, $p_{C_0}$ do not move, we have $x(ABC)=x(A_0B_0C_0)$ as shown in figure \ref{fig:radialcounter1}.  But $x(AB)$ and $x(BC)$ become closer to being null, so their lengths decrease.  As a result, the HRT areas now violate SSA.

Note that since the new endpoints lie to the future of the original null slice, the extremal surfaces and intervals can no longer be placed in a common Cauchy slice. In particular, though $ABC$ remains achronal, $x(ABC)$ is now in the past of $B$. Thus, $x(ABC)$ cannot be a restricted maximin surface. The actual restricted maximin surfaces are harder to identify, but must satisfy SSA by Corollary \ref{cor:ssa} in section \ref{sec:results} below.

\begin{figure}[t]
        \begin{center}
        \begin{subfigure}{.49\linewidth}
        \centering
                \includegraphics[width=0.75\textwidth]{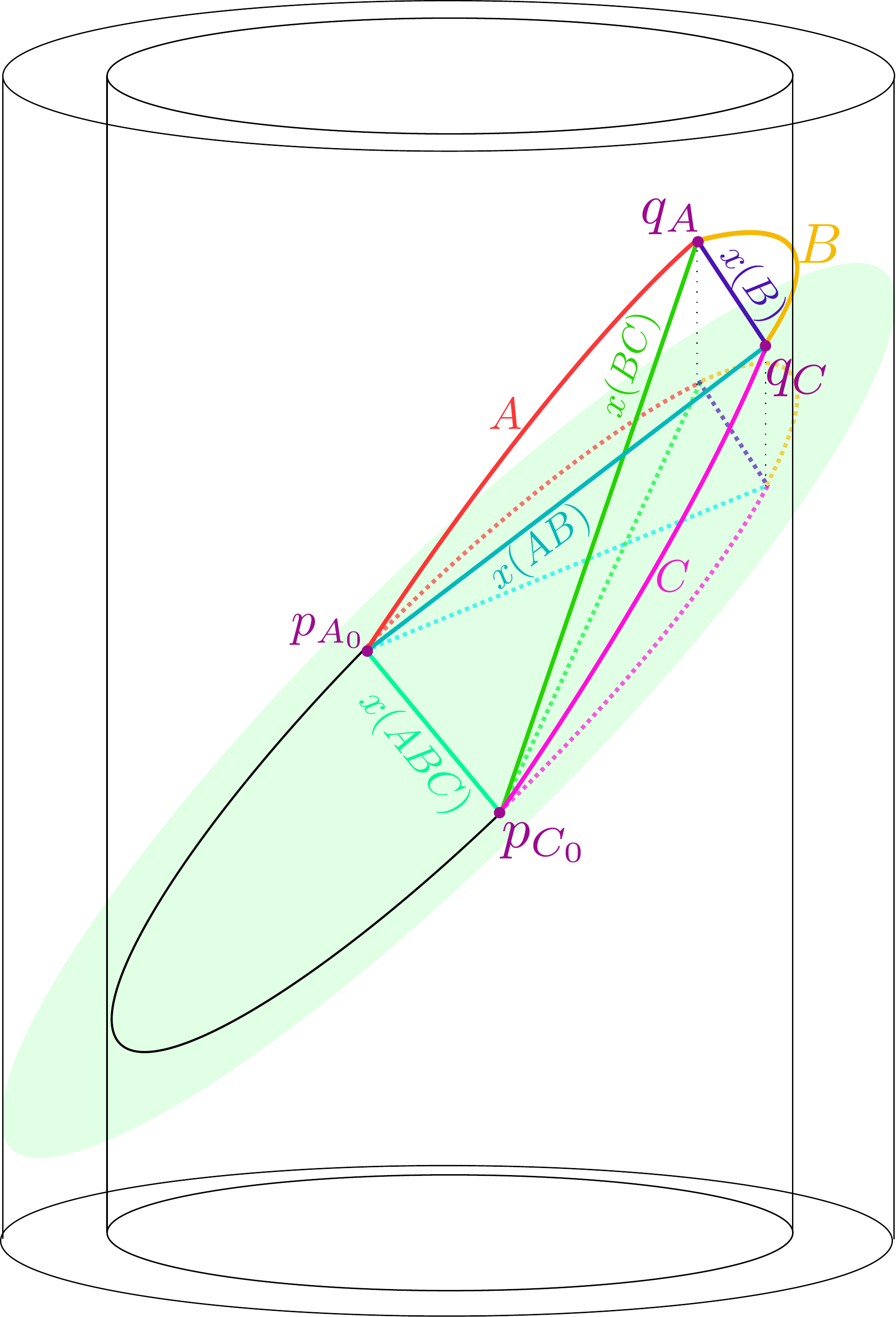}\quad
                \end{subfigure}
		\begin{subfigure}{0.49\linewidth}
        \centering
                \end{subfigure}
        \end{center}
        \caption{An example violating of strong subadditivity for HRT areas with a cutoff in vacuum AdS$_3$. The setting in figure \ref{fig:radialcounter1} has been deformed by acting with a small time translation on the future endpoints $q_{A_0}$ and $q_{C_0}$ of $A_0$ and $C_0$ to define new future endpoints $q_A$, $q_C$, and new intervals $A,B$, and $C$.}
        \label{fig:radialcounter2}
\end{figure}

\subsection{Definitions and Lemmas}
We now state the notation and conventions that will be used in the remainder of the paper.

Throughout this work we will take the bulk spacetime to be classical, smooth. The spacetime is also assumed to satisfy 1) the Null Curvature Condition (NCC), $R_{ab}k^ak^b\geq 0$, for any null vector $k^a$ and 2) a so-called generic condition such that there is nonzero null curvature $R_{ab}k^ak^b$ or shear $\sigma_{ab}\sigma^{ab}$ along any segment of any null curve.  Note, however, that strong subadditivity is a closed inequality (e.g. saturation is allowed) and that the generic condition allows the above curvature or shear to be arbitrarily small. As a result, a proof of strong subadditivity using the generic condition immediately implies that strong subadditivity continues to hold even when the generic condition is not enforced. To see this, one need only approximate the spacetime in which the generic condition fails (but all other assumptions hold) by a sequence of generic spacetimes and take an appropriate limit.

In principle, we would like to allow the cutoff surface $\gamma$ to be any closed codimension 2 achronal submanifold of the bulk spacetime.  We emphasize that achronality is defined by the bulk causal structure. Indeed, we introduce no notion of a codimension 1 boundary. In fact, we could weaken the above condition to also allow portions of $\gamma$ to lie along any asymptotically locally anti-de Sitter (AlAdS) boundary (and so, in particular, we do not require $\gamma$ to be compact). 

While the examples considered above include cutoff surfaces $\gamma$ with sharp corners or null components,  we will take $\gamma$ to be smooth and spacelike for the purposes of making the arguments below. But again, because strong subadditivity is a closed inequality, a proof for smooth spacelike $\gamma$ immediately implies that strong subadditivity holds for e.g. all piecewise smooth, achronal $\gamma$, as such a $\gamma$ can be approximated by the appropriate limit of a sequence of smooth spacelike curves.

The fact that we treat $\gamma$ as a cutoff means that our bulk spacetime is the domain of dependence $D(\gamma)$ of an anchronal surface $\Sigma$ with boundary $\partial \Sigma = \gamma$, so that we effectively work in a  the globally hyperbolic spacetime $D(\gamma)$.  We thus use the term Cauchy surface to refer to $\Sigma$ or any other achronal $\Sigma'$ with the same domain of dependence.  We  use  $\mathcal{C}_{\gamma}$ to denote the set of all such Cauchy surfaces for $D(\gamma)$.

There are various possible subtleties associated with the fact that general Cauchy surfaces need not be smooth.  It is not clear that all such subtleties were explicitly addressed in the original maximin paper \cite{Wall:2012uf}, and we will not attempt to do so here. We will instead assume below that all relevant Cauchy surfaces $\Sigma$ are at least piecewise smooth and leave treatment of the more general case for future work.   Note that the piecewise smooth case allows points $p$ on $\Sigma$ where the space of tangent vectors at $p$ depends on the direction from which $p$ is approached.  We expect this generalization of the smooth case to be important, as our maximin procedure may give rise to surfaces $\Sigma$ that partially coincide with the boundary $\partial D(\gamma)$.  Failures of smoothness will certainly occur at caustics of the null congruences along this boundary, and they may also arise when $p$ lies at the boundary of $\Sigma \cup \partial D(\gamma)$.  We will similarly assume below that any maximin surface is piecewise smooth. 

As one may expect, we will make significant use of the ingoing and outgoing future pointing null congruences orthogonal to $\gamma$. We denote their affinely parametrized tangent vectors respectively by $k^a$ and $l^a$ respectively, with $\theta_k$ and $\theta_l$ the corresponding null expansions. 

Having established this notation and the above conventions, we now define our restricted maximin surfaces in two steps.

\begin{defn} For any subregion $A\subset\gamma$, and for each element $\Sigma \in C_\gamma$, let min$(A,\Sigma)$ be the codimension 2 surface in $\Sigma$ which is anchored to $\partial A$, homologous to $A$ within $\Sigma$, and has minimal area consistent with the above constraints. If there are multiple such surfaces, min$(A,\Sigma)$ can refer to any of them.
\end{defn}

\begin{defn}The \textit{restricted maximin surface} $\Mm(A)$ is then the minimal surface min$(A,\Sigma)$ whose area is  maximal with respect to variations over surfaces $\Sigma \in C_\gamma$. By contrast, $x(A)$ will by used to denote the smallest extremal surface anchored to $\partial A$.
\end{defn}

We also mention the following two lemmas that will be used in  the next section.

\begin{lem}
The boundary of $D(\gamma)$ is $\partial D(\gamma)=L^{+}\cup L^{-}\cup \gamma$ where $L^{+}$ is the set of points $p \notin \gamma$ that are reached by null geodesics along the vector field $k$ which start at $\gamma$ and which have not arrived at any conjugate point or nonlocal geodesic intersection  before reaching  $p$.   $L^{-}$ is defined similarly with $k$ replaced by $-l$. Note that $L^{\pm}$ includes points on caustics.
\end{lem}

This  lemma is precisely Theorem 1 of \cite{Akers:2017nrr} restated in our particular context using the above notation.  As a result, it follows immediately from their argument.

\begin{lem}
Suppose $N_1$ and $N_2$ are two smooth null congruences that are tangent at some point $p$ on a Cauchy surface $\Sigma$. If $N_2$ is nowhere to the (chronological) past of $N_1$, then in any sufficiently small neighborhood of $p$, either i) $N_1$ and $N_2$ coincide, or ii) there exists a point $y$ at which $\theta_{N_2} >\theta_{N_1}$.  Here we may use any smooth map between $N_1$ and $N_2$ to compare points on the two surfaces.
\label{lemma:null}
\end{lem}

This lemma is Theorem 4 of \cite{Wall:2012uf}. Because we will often wish to study congruences which are not obviously smooth, we pause to state the following immediate corollary, which we number zero for later convenience.

\begin{cor} Note that both cases (i) and (ii) of Lemma \ref{lemma:null} require $\theta_{N_2}(x) \ge\theta_{N_1}(x)$.  Indeed, equality is manifest in case (i), and otherwise $\theta_{N_2}(y) > \theta_{N_1}(y)$ in sufficiently small neighborhoods of $x$. Since the inequality $\theta_{N_2}(x) \ge\theta_{N_1}(x)$ is closed, it must also hold whenever $N_1,N_2$ can be approximated by pairs of smooth congruences that are tangent at $x$, even if $N_1,N_2$ are not smooth themselves.  When $\theta_{N_2}(x)$ or $\theta_{N_1}(x)$ is ill-defined, we understand the inequality to hold for all smooth approximating congruences\footnote{Alternatively, one may treat the corresponding null expansions as having delta-function-like terms at places where piece-wise smooth codimension 2 surfaces from which they are orthogonally launched fail to be $C^1$.}.  This allows one to deal with the case where $N_1,N_2$ have caustics or non-local intersections at $x$. And this in turn means that we can apply the result to congruences launched orthogonally from codimension 2 surfaces that are only piecewise smooth, and perhaps in worse cases as well.
\label{cor:bend}
\end{cor}

\smallskip


\begin{figure}[h]
\centering
\includegraphics[width = 0.5\linewidth]{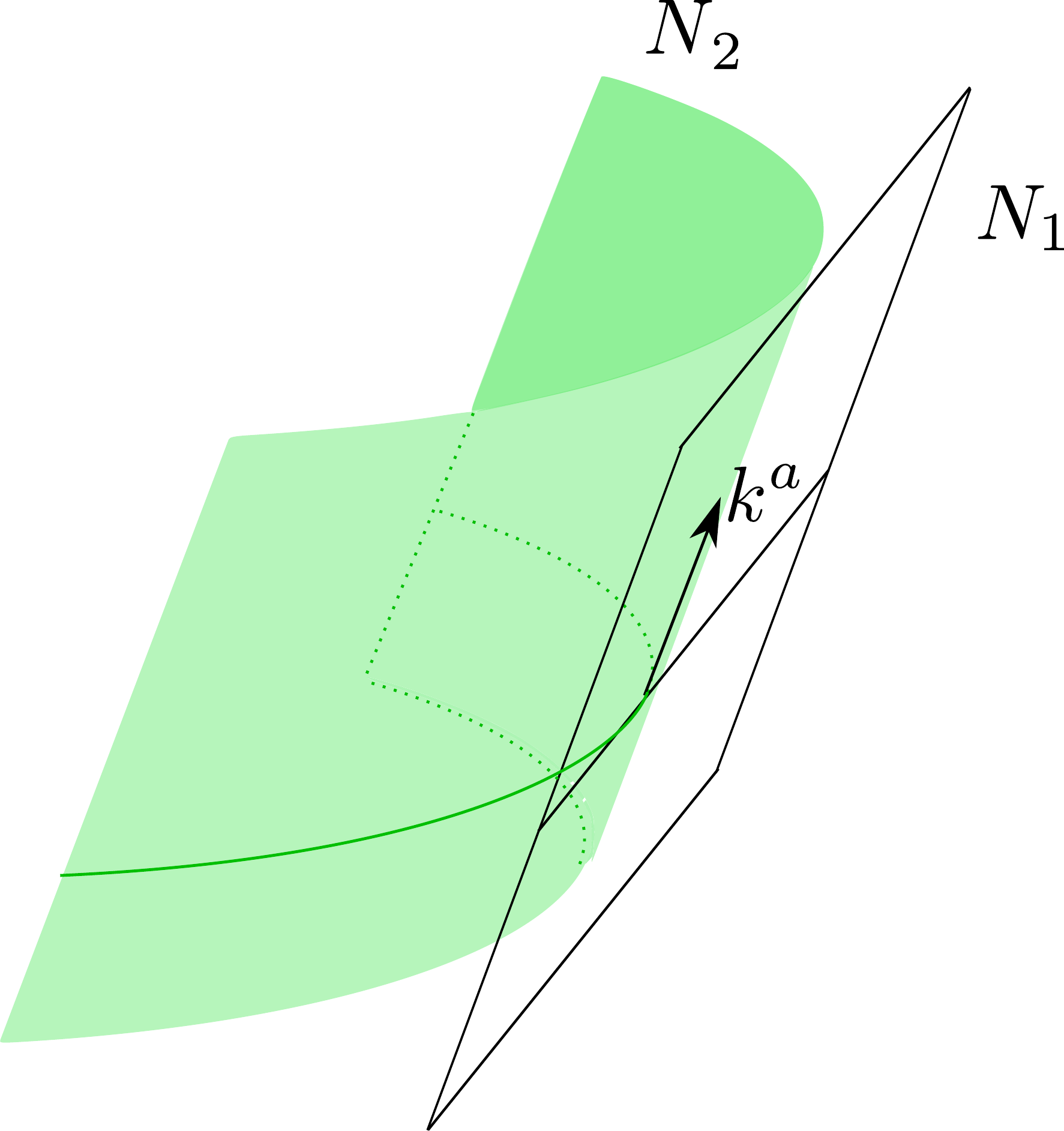}
\caption{The lines depict spacelike cuts of null congruences $N_1$ and $N_2$, tangent at a point. $N_2$ is nowhere to the past of $N_1$ and is thus expanding faster by Lemma \ref{lemma:null}}
\label{fig:nullfol}
\end{figure}

As a final preliminary remark, we recall that \cite{Wall:2012uf} suggested that one may always choose a maximin surface that is stable in the following sense: when a Cauchy surface $\Sigma$  on which $\Mm(A)$ is minimal is deformed to a nearby Cauchy surface $\Sigma_\epsilon$, the minimal surface on $\Sigma_\epsilon$ should remain close to $\Mm(A)$.  Indeed, it was suggested that when $\Mm(A)$ is unique it is always stable and, when it is not, that at least one of the allowed choices will be stable (and one should make such a choice).
When $\Mm(A)$
is both spacelike and smooth, this follows by the technical argument in
section 3.5 of \cite{Wall:2012uf}.
But more generally the proof is incomplete.
Below, we follow both \cite{Wall:2012uf} and \cite{Marolf:2019bgj} in simply assuming this property to hold and studying its implications.

\section{Main Results}
\label{sec:results}

In deriving the desired results, we will follow the same basic proof strategy as in \cite{Wall:2012uf}.  The first ingredient in this strategy is the realization that the expansions of null congruences orthogonal to maximin surfaces $\Mm(A)$ are significantly constrained.  In particular, Claim \ref{claim:ext} below will establish that at any point $p \in \Mm(A)$ with $p \not \in \gamma$, null congruences orthogonal to $\Mm(A)$ that do not immediately leave $D(\gamma)$ have non-positive expansion.  For example, if $p \in L^+$, then it can have positive expansion $\theta_l$, but the expansions $\theta_k, \theta_{-k}, \theta_{-l}$ must be non-positive.  In particular, $\theta_k = -\theta_{-k}$ will vanish.

The second ingredient is the notion of joint (restricted-) maximin surfaces.  Such surfaces are defined by choosing arbitrary $c,d>0$ and considering the quantity $Z = c \text{Area(}X_{A})+d \text{Area(}X_{B})$.  In particular, for each $\Sigma' \in C_\gamma$, we define the surfaces $\min(Z,\Sigma')_A$ and $\min(Z,\Sigma')_B$ to be the surfaces that minimize $Z$, subject to the constraint that $X_{A}$ and $X_B $ are codimension 2 surfaces in $\Sigma'$ anchored respectively to $\partial A$ and  $\partial B$, and satisfying associated homology constraints on $\Sigma'$.  Note that since $c,d >0$, the surfaces
$\min(Z,\Sigma')_A$ and $\min(Z,\Sigma')_B$ must each be separately minimal on $\Sigma'$.  So for all such $c,d$, we have $\min(Z,\Sigma')_A = \min(A,\Sigma')$ and $\min(Z,\Sigma')_B = \min(B,\Sigma')$. 
The joint restricted-maximin surfaces $Y_A$ and $Y_B$ are then defined by maximizing $Z$ with respect to variations of $\Sigma'$. We assume that $Y_A$ and $Y_B$ are stable in parallel with our assumption regarding $\Mm(A)$.

By construction, the joint restricted-maximin surfaces are both minimal on the same slice $\Sigma$.  As a result, for such surfaces we can use the minimal surface argument of \cite{Headrick:2007km} to establish SSA.  Indeed, while SSA involves three regions of $\gamma$, and thus three surfaces in the bulk, it turns out to be enough to work with pairs $A,B$ of regions of $\gamma$ that are nested in the sense that $A \subset B$. As a result, SSA for the original restricted maximin surfaces will follow if we can show that they agree with joint restricted-maximin surfaces for nested pairs $B \supset A$.

An important component of doing so is to show that, at least when $B \supset A$, the expansions of null congruences orthogonal to $Y_A,Y_B$ obey precisely the same constraints as those described above for null congruences orthogonal to $\Mm(A)$ and $\Mm(B)$.  For clarity, we state such results for $Y_A,Y_B$ below as corollaries to the corresponding claims for $\Mm(A), \Mm(B)$, explaining any additional relevant details.
 
With the above-mentioned results about null expansions in hand,  
at least when $B \supset A$, one can indeed establish that $Y_A = \Mm(A)$ and $Y_B = \Mm(B)$.  The main idea is to use the above results for null expansions to show that failure to coincide is inconsistent with $\Mm(A)$ or $\Mm(B)$ being minimal on any Cauchy slice $\Sigma$.  A key tool in such comparisons is the notion  introduced in \cite{Wall:2012uf} of the {representative} $\tilde{y}(\Sigma)$ on a Cauchy slice $\Sigma$ of a codimension 2 achronal surface $y$.  The precise definition we use largely follows the presentation in \cite{Marolf:2019bgj}, but fixes issues associated with the fact that Cauchy surfaces $\Sigma$ become null:

\begin{defn} We begin with a spacetime-codimension 2 achronal surface $y$ lying in a Cauchy slice $\Sigma'$, and which is homologous within $\Sigma'$ to some boundary region $A$. (Note that since $\Sigma'$ is codimension 1 in the spacetime, $y$ is codimension 1 with respect to $\Sigma'$.)
Given another Cauchy surface $\Sigma$, representatives $\tilde y(\Sigma)$  are defined by observing that $y$ splits the original slice $\Sigma'$ into two pieces: $\Sigma'_A$ (with boundary $y \cup A$) and $\Sigma'_{\bar A}$ (with boundary $y \cup \bar A$).  Let the associated domains of dependence be $D(\Sigma'_A)$ and $D(\Sigma'_{\bar A})$ with boundaries $\partial D(\Sigma'_A)$ and $\partial D(\Sigma'_{\bar A})$.  A representative $\tilde y(\Sigma)$ on $\Sigma$ of $y$ is a codimension 2 surface homologous to $A$ within $\Sigma$ that contains $y \cap \gamma$ and lies in both $\Sigma$ and
$\partial D(\Sigma'_A) \cup \partial D(\Sigma'_{\bar A})$.   Note that at least one representative must exist on any $\Sigma$ since both $\Sigma$ and $ \Sigma'$ are Cauchy surfaces for $D_\gamma$ and thus both contain $\gamma$ itself.  In particular, when $\Sigma$ is spacelike, two possible choices of representative are $\Sigma \cap \partial D(\Sigma'_A)$ and $\Sigma \cap \partial D(\Sigma'_{\bar A})$. One may sometimes wish to state which such representative one chooses, but in fact the choice of representative will not matter for our purposes below.
\end{defn}

This definition becomes useful when combined with the above results about null expansions.
Since $\tilde y(\Sigma) \subset \partial D(\Sigma'_A) \cup \partial D(\Sigma'_{\bar A})$, each point on $y(\Sigma)$ can be reached from a distinct point on $y$ by following a generator of a null congruence orthogonal to $y$ that remains in $D(\gamma)$, and where the generator is free of conjugate points or non-local self-intersections.  If $y$ is a connected component of a restricted maximin surface or a joint restricted-maximin surface then, and if $y$ itself does not lie in $\Sigma$, the generic condition and the above-mentioned results from Claim \ref{claim:ext} and Corollary \ref{cor2} guarantee the area of $\tilde y(\Sigma)$ to be strictly smaller than that of $y$.  Here it is important that $\tilde y(\Sigma)$ contains $y \cap \gamma$ so that, at such points (where the null expansions are not controlled), the null geodesic is followed only for vanishing affine-parameter distance.

One can now quickly conclude that $Y_A = \Mm(A)$.  Since both surfaces are of the form $\min(\Sigma',A)$, the maximization step of the maximin procedure implies $\A(\Mm(A))\ge \A(Y_A)$.  But if $Y_A$ does not coincide with $\Mm(A)$, then we may choose the Cauchy surface $\Sigma$ on which $\Mm(A)$ is minimal so that at least part of $Y_A$ does not lie on this surface\footnote{The stability condition implies that if $\Mm(A)$ is minimal on some $\Sigma_0$, then it is also minimal on all nearby slices $\Sigma_\epsilon$ that coincide with $\Sigma_0$ on some open set around $\Mm(A)$. So since maximality forbids $Y_A \subset \Mm(A)$ with $Y_A\neq \Mm(A)$, for $Y_A\neq \Mm(A)$ the condition will hold for some $\Sigma_\epsilon$.}.  The representative $\widetilde{Y_A}(\Sigma)$ on $\Sigma$ of $Y_A$ must then have smaller area than $Y_A$, contradicting minimality of $\Mm(A)$ on $\Sigma$. The identical argument also proves that $Y_B$ coincides with $\Mm(B)$.  Entanglement wedge nesting then follows from the fact that $Y_A,Y_B$ are minimal on the same slice $\Sigma$ and thus cannot cross.  Strong Subadditivity and Monogamy of mutual information follow as well.

This completes our summary of the overall proof strategy.
As in \cite{Marolf:2019bgj}, the full proofs below repeatedly use the fact that portions of $\Mm(A)$ lying entirely in the interior of $D(\gamma) $ will behave much like the maximin surfaces of \cite{Wall:2012uf}. As a result, we typically proceed by studying various cases, depending on whether relevant points lie in the interior, on $L^+$ or $L^-$, or on $\gamma$.  When dealing with $L^\pm,$ we will focus on studying $L^+$ with the understanding that analogous results immediately  follow for $L^-$. Furthermore,  even in dealing with issues associated with the boundary of $D(\gamma)$, the arguments below largely follow the structure of the derivations in \cite{Wall:2012uf}.

Before proceeding, we remind the reader that we will treat every relevant Cauchy surface and every restricted maximin surface as being piecewise smooth.  Proving this to be the case, or showing that the results hold more generally is an important open issue that is beyond the scope of this work.

\subsection{Details of the argument}

We now fill in the details of the argument sketched above, breaking the derivation into three separate claims regarding $\Mm(A)$, together with two corollaries for joint restricted-maximin surfaces.  The third claim establishes entanglement wedge nesting, and the other main results of SSA and Monogamy of Mutual Information are then an additional corollary.

\begin{claim}
\label{claim:spacelike}
Suppose $p, q \in \Mm(A)$, with $p\neq q$ and neither $p$ nor $q$ in $\gamma$. Then $p, q$ are spacelike separated.
\end{claim}

\begin{proof}

Since $p,q \in \Mm(A)$, the points lie in a common Cauchy surface.  This means that they cannot be connected by any timelike curve.  So to establish the desired result, we need only exclude the possibility that they might be connected by a null curve.
The argument proceeds in two steps:  (a) We first show that $p,q$ cannot be connected by a null curve $\eta$ lying in $\Mm(A)$.  (b)  We then show that $p, q$ cannot be connected by any null curve $\eta$.

(a) As argued in \cite{Wall:2012uf}, (unrestricted) maximin surfaces cannot become null at any point $y$.  This can be seen by realizing that, if they were null at some $y$, then the Cauchy slice containing the surface could be deformed slightly in a neighborhood of $y$ to be everywhere spacelike.  Stability of maximin surfaces then requires the minimal surface in the new Cauchy surface be nearby.  The fact that a spacelike piece of the new surface is obtained by deforming a null and nearly null piece of the maximin surface means that its area is greater than that of the maximin surface at first order in the deformation.  But this contradicts the fact that the maximin surface is maximal with respect to variations in the Cauchy slice. Thus there can be no point $y$ at which the unrestricted maximin surface becomes null, and the surface can contain no null curves.

This argument depends only on variations in a neighborhood of $y$, and thus carries over unchanged to our restricted maximin construction if any portion of the null curve $\eta \in \Mm(A)$ lies in the interior of $D(\gamma)$.  But in fact the argument also holds when the case where  $\eta$ lies in $\partial D(\gamma)$, as the Cauchy surface can still be deformed to a spacelike slice using a diffeomorphism generated by a vector field that points into the interior on  $\partial D(\gamma)$.

(b) We now extend this result to null curves $\eta$ that are not contained in $\Mm(A)$.  We argue by contradiction, assuming that $\eta$  exists and, without loss of generality, taking $p$ to lie to the future of $q$.  We consider in detail the case where $p$ and $q$ lie in $L^+$, so any null curve connecting them must be a null generator of $L^+$. Other cases follow similarly, or by arguments that are even more nearly identical to those of Thm 14 of \cite{Wall:2012uf}. As in \cite{Wall:2012uf}, the proof is unchanged if we take $p$ and $q$ to be contained in distinct members of a pair of joint restricted-maximin surfaces $Y_A, Y_B$.  The latter case (which we list a separate corollary below) will be critical in the proof of Claim \ref{claim:ewn} below. To make it clear that the argument also applies in that context, in the remainder of this proof we will write $q\in Y_B$ and $p\in Y_A$, where $Y_A$ and $Y_B$ can be either different joint restricted-maximin surfaces or just different subsets of $\Mm(A)$.

Since $p \in L^+$, it lies on some null generator of $L^+$.  Because $\Sigma$ is an achronal surface containing $p$ and $\gamma$, $\Sigma$ must also contain the part of that null generator lying to the past of $p$.  As in figure \ref{fig:no_null}, we can compare the future-directed ingoing null congruence $N_B$ orthogonal to $Y_B$ with the future-directed ingoing null congruence orthogonal to  $Y_A$. Note that $N_B$ must contain $p$ since $N_B\subset D(\gamma)$ and the only such future-directed  null geodesic at $q \in L^+$ which remains in $D(\gamma)$ is the generator of $L^+$ that connects $q$ with $p$.

Now, since $Y_A$ and $Y_B$ lie on the same achronal slice, $Y_B$ cannot enter the chronological past of $Y_A$.  Similarly, the part of $N_B$ to the future of $Y_B$ cannot enter the chronological past of $Y_A$, and must in fact be nowhere to the past of $N_A$ near $p$.  We may thus apply Corollary \ref{cor:bend} from Section \ref{sec:prelim}, so that $\theta_k(B,p)\ge\theta_k(A,p)$. The generic condition then implies that either $\theta_k(B,q)>0$ or $\theta_k(A,p)<0$.

To proceed, we introduce the extrinsic curvature one-form $\mathbf{tr(K^B)} \equiv K^B_i dx^i$ of $Y_B$. Since $\Sigma$ contains the entire null generator of $L^+$ to the past of $p$ (down to, but not including $\gamma$), it must be
smooth along $k$ at $q$. Using the fact that  $Y_B$ is minimal on $\Sigma$ then yields $K^B_ik^i=0$ at $q$, and thus $\theta_k(B,q)=0$.

Our dichotomy above then requires $\theta_k(A,p)<0$, and thus that $\theta_{-k}(A,p)>0$.  The latter refers to the past-directed null congruence. We may also say that $-k$ points outward (towards $Y_B$) along $\Sigma$ from $Y_A$.  But $Y_A$ will also have an inward pointing normal in $\Sigma$, which we denote $v$. (If $\Sigma$ is only piecewise smooth then $v$ can differ from $k$.) Since $Y_A$ is minimal on $\Sigma$, we have $K^A_i v^i \geq 0$. Note the the codimension 2 extrinsic curvature $K^A_i$ of $\Mm(A)$ is well-defined\footnote{If $\Mm(A)$ is only piece-wise smooth, one may consider a family of smooth approximating surfaces.} even if $\Sigma$ fails to be smooth at $\Mm(A)$. Since the above null expansion requires $-K^A_i k^i > 0$, the vectors  $-k^i$ and $v^i$ cannot be diametrically opposed so, to remain tangent to the achronal surface $\Sigma$, $v^i$ must point in a spacelike or null past-inward direction; see again figure \ref{fig:no_null}.

We can now consider deforming $\Sigma$ to the past infinitesimally at $p$. Since the area of $Y_A$ (or a positively weighted sum involving the area of $Y_A$) is maximal with respect to variations of $\Sigma$, there must be some vector $u^i$ at $p$ that is orthogonal to $Y_A$, points to the past of $\Sigma$, and which has $K^A_i u^i \leq 0$. But since $Y_A$ has spacetime codimension two, this $u^i$ must be a linear combination of $v^i$ and $-k^i$. Furthermore, since $u^i$ points to the past of $\Sigma$, it must be a  {\it positive} linear combination of these vectors (see again figure \ref{fig:no_null}).  But this implies $K^A_iu^i >0$,  contradicting the conclusion above.

By contrast, note that  if $q\in \gamma$ then $\theta_k(B,q)$ need not be strictly zero (as minimality on $\Sigma$ only requires $\theta_k(B,q)\geq 0$). And if $\theta_k(B,q) > 0$, then $\theta_k(A,p)$ may also be positive. Thus, we cannot rule out possible null separations between $p$ and $q$ if $q\in \gamma$ (as consistent with the statement of the claim).

\begin{figure}
    \centering
    \includegraphics[width = 0.7\linewidth]{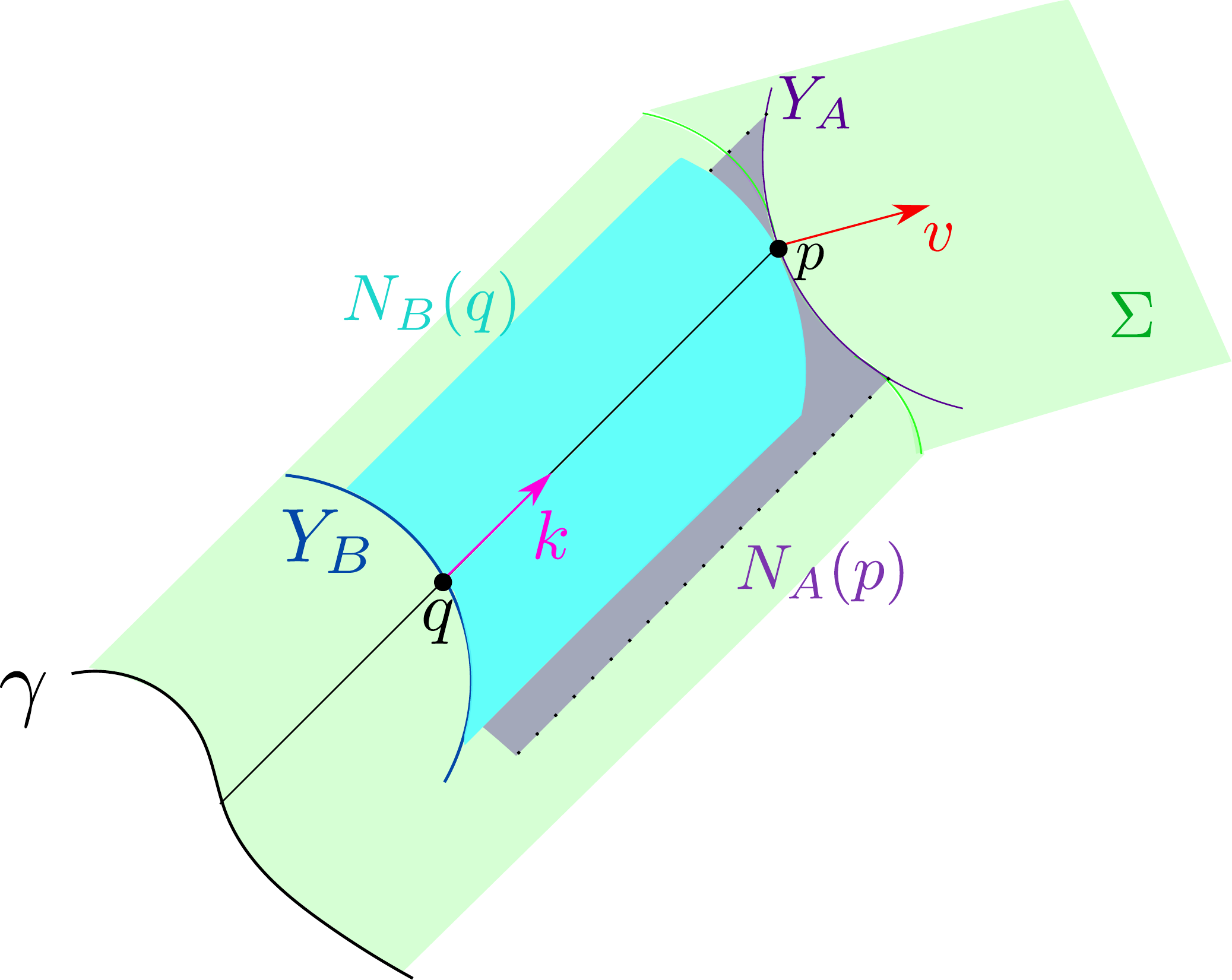}
    \caption{Consider two surfaces $Y_A$ and $Y_B$ as in Claim \ref{claim:spacelike} which are constructed on the same slice $\Sigma$. Points $q$ on $Y_B$ and $p$ on $Y_A$ cannot have null separated points. $N_B(q)$ is the null congruence orthogonal to $Y_B$ at $q$ and $N_A(p)$ is the null congruence orthogonal to $Y_A$ at $p$. This null congruence must bend towards the orthogonal vector $v$, because $Y_A$ must stay contained in $\Sigma$ (and so remain spacelike separated to $Y_B$). This places constraints of the relative signs of the expansions of the null congruences, which can in turn be used to contradict the maximality of the areas $Y_A$ and $Y_B$ along $\Sigma$.}
    \label{fig:no_null}
\end{figure}

\end{proof}

\begin{cor}
\label{cor1}
Consider the joint restricted-maximin surfaces $Y_A,Y_B$ defined by choosing $c,d>0$ and performing restricted maximin using quantity $Z = c \text{Area(}X_{A})+d \text{Area(}X_{B})$.  Here we require that $X_{A}$, $X_B $ are codimension 2 surfaces anchored respectively to $\partial A$ and  $\partial B$ and satisfying associated homology constraints on each Cauchy slice. As described above, this results in a pair of surfaces $Y_A,Y_B$ that are minimal on the same slice $\Sigma$.  Let $p \in Y_A$ and $q \in Y_B$. If neither $p$ nor $q$ are in $\gamma$, then either $p=q$ or $p, q$ are spacelike separated.
\end{cor}

The argument is identical to that for Claim 1.  This corollary will be useful below in much the same way as in \cite{Wall:2012uf}.

\begin{claim}
\label{claim:ext}
Suppose $p\in\Mm(A)$ with $p\notin\partial D(\gamma)$. Then, $\Mm(A)$ is extremal at $p$. If $p\in L^\pm$, then $\theta(p) \leq 0$ for any null congruence orthogonal to $\Mm(A)$ that is directed into the interior of $D(\gamma)$.  Furthermore, $\theta_k=0$ for $p\in L^+$ and $\theta_l =0$ for $p \in L^-$.  As a result, at any point $p \in \Mm(A)$ with $p \not \in \gamma$, null congruences orthogonal to $\Mm(A)$ that do not immediately leave $D(\gamma)$ have non-positive expansion.
\end{claim}

\begin{proof}
The analogous result for unrestricted maximin surfaces was proven in \cite{Wall:2012uf}. Since the arguments are local at $p$, they also establish Claim \ref{claim:ext} for $p \notin \partial D(\gamma)$. By similar reasoning, even if $p\in L^+$  minimality of $\Mm(A)$ on $\Sigma$ implies that $\Mm(A)$ is extremal along directions within $\Sigma$. 

Note that the arguments of Claim \ref{claim:spacelike} above show that $\Sigma$ must be smooth at $p$. Otherwise, minimality along the two tangent vectors of $\Sigma$, $-k$ and $v$, would imply that the area of $\Mm(A)$ could be increased by varying $\Sigma$ to the past at $p$, contradicting maximality of $\Mm(A)$. 

Since $\Sigma$ is smooth at $p$, it must be tangent to the $k$ direction at $p$. Minimality of $\Mm(A)$ on $\Sigma$ then gives $K^A_ik^i=0$, and hence $\theta_k=0$.

We can now consider the null expansion in the $l$ direction. In the maximization step of the maximin procedure, $\Sigma$ can be freely varied in the $l$ direction into the interior of $\partial D(\gamma)$ at points near $p$. Maximality over these variations implies $K_i^A l^i\leq 0$, which in turn implies $\theta_l \leq 0$. Of course, similar arguments hold for $p\in L^-.$
\end{proof}

Again, the argument generalizes to yield analogous results for the joint restricted-maximin surfaces $Y_A,Y_B$ defined as in Corollary \ref{cor1} by performing a restricted maximin procedure on  $Z = c \text{Area(}X_{A})+d \text{Area(}X_{B})$.  However, this time the generalization requires a bit more work as explained below:
\begin{cor}
\label{cor2}
Let $Y_A,Y_B$ be defined as in Corollary \ref{cor1} for nested regions $B\supset A$, and consider $p\in Y_A \cup Y_B$ with $p\notin\partial D(\gamma)$. Then, $Y_A \cup Y_B$ is locally an extremal surface at $p$.  In particular, if $Y_A$ and $Y_B$ coincide at $p$ then they must do so in a neighborhood of $p$.
Furthermore,  if $p\in L^\pm$, then $\theta(p) \leq 0$ for any null congruence emanating from $Y_A \cup Y_B$ that is directed into the interior of $D(\gamma)$.  We also have $\theta_k=0$ for $p\in L^+$ and $\theta_l =0$ for $p \in L^-$.  As a result, at any point $p \in \Mm(A)$ with $p \not \in \gamma$, null congruences orthogonal to $Y_A$ or $Y_B$ that do not immediately leave $D(\gamma)$ have non-positive expansion.
\end{cor}

\begin{proof}
For $p\notin \partial D_\gamma$, the key fact is that  minimal surfaces on the same slice $\Sigma'$ cannot touch, cross, or coincide, except in the case where
an entire connected component of $\min(A,\Sigma')$ coincides with an entire connected component of $\min(B,\Sigma')$.  At a point $p$ where the surfaces coincide, maximizing one area under local variations of the Cauchy surface also clearly maximizes the other\footnote{The fact that variations of $\Sigma$ may cause $Y_A, Y_B$ to separate from each other can affect the areas only at second order since both surfaces are minimal on the original Cauchy slice, and since we now consider only interior points where such minimal surfaces must be smooth.}. So extremality of $Y_A$ and $Y_B$ at such points follows from the same argument as for $\Mm(A)$.   Furthermore,
Corollary \ref{cor1} states that (except on full connected components where $Y_A$,  $Y_B$ coincide),
points on $Y_A$ are spacelike separated from points on $Y_B$. On portions where they do not
coincide, the stability requirement implies that we may then freely vary the Cauchy slice in the neighborhood of one surface without affecting the
area of the other.  So again extremality of $Y_A,Y_B$ follows in precisely the same way as for $\Mm(A)$.

It remains only to consider $p\in L^\pm$.  We focus on the case $p \in L^+$, with the understanding that the time-reversed remarks hold for $L^-$.  There are three cases to consider below.

\indent \textit{Case 1}: Suppose $p\in L^+$ lies in $Y_A$ but not in $Y_B$.
Recall that Corollary \ref{cor1} states that no two points on these surfaces can be null separated if neither one is on $\gamma$. (They cannot be separated by a null ray that runs through the bulk, nor by a null generator contained in $\partial D(\gamma)$). Because minimal surfaces on the same slice cannot cross, if $Y_A$ has a point on some generator $k$ of $L^+$, then $Y_B$ must have a point $q\in \gamma$ along the same generator.   Since $q \in \gamma$, it  lies in all Cauchy surfaces and changes in $Z$ resulting from local variations of the Cauchy surface near the generator $k$ receive no contributions from $Y_B$.  Applying the arguments of Claim \ref{claim:ext} to such variations then establish the desired results for $Y_A$ (and the results for $Y_B$ hold vacuously).

\indent \textit{Case 2}: Suppose $p\in L^+$ lies in $Y_B$ but not in $Y_A$.
Then we may again use Corollary \ref{cor1} to show that if $Y_B$ has a point in $L^+$ on some generator $k$ then $Y_A$ cannot intersect that generator at all.  As a result, one is free to vary the Cauchy surface near $k$ without affecting the area of $Y_A$.  Applying the arguments of Claim \ref{claim:ext} to such variations then proves the desired results for $Y_B$ (and the results for $Y_A$ hold vacuously).

\indent \textit{Case 3}: Finally, suppose that $Y_A$ and $Y_B$ coincide at $p\in L^+$.
Below, we will show that $Y_A$ and $Y_B$ cannot just touch at such a point $p$, but must in fact must coincide on finite regions whose boundaries must lie in $\gamma$.  
Applying the arguments of Claim \ref{claim:ext} then shows
that a weighted sum of the $A,B$ null expansions at $p$ will satisfy the claim.  But since both surfaces define the same null congruences and null expansions near $p$, the same must be true for the expansions of the individual null congruences orthogonal to $Y_A$ and $Y_B$.

The argument for coincidence of $Y_A$ and $Y_B$ on the above finite region will roughly follow the proof of Claim \ref{claim:spacelike}, though now $p$ and $q$ coincide. As above, we can compare the $k$ directed null congruence $N_B(p)$ from $Y_B$ at $p$ to the $k$ directed null congruence $N_A(p)$ from  $Y_A$ at $p$. As in Claim \ref{claim:spacelike}, because $Y_A$ and $Y_B$ lie on the same achronal slice, because minimal surfaces on the same slice cannot cross, and because $A\subset B$, we must have  $N_A$ nowhere to the past of $N_B$. Corollary \ref{cor:bend} then requires $\theta_k(B,p) \geq \theta_k(A,p)$.

We further consider the extrinsic curvature $\mathbf{tr(K^B)} \equiv K^B_i dx^i$ of $Y_B$ and the analogous $\mathbf{tr(K^A)} \equiv K^A_i dx^i$. Note that $\Sigma$ must contain the full generator of $L^+$ to the past of $p$. Because $Y_A$ and $Y_B$ are minimal on $\Sigma$, we have $K^B_i (-k^i) \geq 0$ and $K^A_i(-k^i)\geq 0$, where $-k^i$ is the past-directed tangent vector to $\Sigma$ which points along the generators of  $L^+$ at $p$. Thus, $\theta_k(B,p) \leq 0$ and $\theta_k(A,p) \leq 0$.

Combining the two preceding paragraphs then requires  either $\theta_k(A,p)=\theta_k(B,p)=0$ or $\theta_k(A,p)<0$.  In the latter case, we can continue to follow the arguments of Claim \ref{claim:spacelike}. $\Sigma$ will possess some other tangent vector normal to $Y_A$, which points to the interior. We call this vector $v_i$. Minimality on $\Sigma$ implies $K^A_i v^i \leq 0$. However, because $\theta_k(A,p)<0$, $K^A_i (-k)^i > 0$, Thus, $v^i$ cannot point along the $k$ direction, and $v^i$ must point in a spacelike or null past direction into the interior of $Y_A$. We can now consider pushing $\Sigma$ to the past infinitesimally at $p$. Since $Y_A$ and $Y_B$ have maximal area with respect to variations of $\Sigma$, pushing $\Sigma$ to the past, along some vector $u^i$ must decrease the area of $Y_A$ and $Y_B$, such that $K_i u^i \leq 0$. But $u^i$ must be a positive linear combination of $v^i$ and $-k^i$, implying $K_iu^i >0$ and contradicting the previous line.
Thus, $\theta_k(A,p)=\theta_k(B,p)=0$.

We now consider the set ${\cal C}$ on which the surfaces coincide, and in particular the connected component ${\cal C}_0 \subset {\cal C}$ that contains $p$.  Since the surfaces are continuous,  ${\cal C}_0$ must be closed. Thus any point $q$  in the boundary of ${\cal C}_0$ must also lie in ${\cal C}$.  But if $q$ lies in the interior of $D(\gamma)$, then we have $\theta_k(A)=\theta_k(B)=0$ at all points near $q$ (as these also lie in the interior). Lemma \ref{lemma:null} then requires our surfaces to coincide on a larger set, contradicting the fact that $q$ lies at the boundary of ${\cal C}_0$.

Similarly, if $q\in L^+$, then all nearby points are in $L^+$ or the interior.  So by the arguments of the above cases $\theta_k$ for each surface must vanish on a larger set. Thus Lemma \ref{lemma:null} again requires the surfaces to coincide on a larger set and yields a contradiction.  The case
$q
\in L^-$ is similarly forbidden by the time-reversed argument.
The only remaining possibility is that $\partial {\cal C}_0 \subset \gamma$ as claimed.  This completes the proof of Corollary \ref{cor2}. 
\end{proof}


\smallskip

As described in the introduction to this section, the above results are sufficient to guarantee that for $A\subset B$ the joint restricted-maximin surfaces $Y_A,Y_B$ agree with $\Mm(A)$ and $\Mm(B)$.  It then follows that  $\Mm(A)$ and $\Mm(B)$ lie on a common Cauchy surface $\Sigma_0$, on which they are both minimal.  Nesting of the associated entanglement wedges then follows from the fact that minimal surfaces on $\Sigma_0$ cannot cross.  We use the notation $EW(A)$ to denote the entanglement wedge of $A$ and thus write $EW(A) \subset EW(B)$.  As usual, we use bulk causality to covariantly define $EW(A)$.  Thus we have:

\begin{claim}
If $A$ and $B$ are regions with $A\subset B$,  $\Mm(A)$ and $\Mm(B)$ are minimal on a common Cauchy slice $\Sigma \in C_{\gamma}$.  This also implies $\text{EW}(A)\subset\text{EW}(B).$ 
\label{claim:ewn}
\end{claim}

\begin{proof}
The full proof was given in the introduction to this section.  Note that as opposed to the original maximin surfaces of \cite{Wall:2012uf}, for boundary regions $A \subset B$ the intersection $\Mm(A) \cap \Mm(B)$ can be very general.  But as in the proof of Corollary \ref{cor2}, except when $\Mm(A)$ and $\Mm(B)$ coincide on finite regions, non-trivial intersections of $\Mm(A)$ with $\Mm(B)$ are confined to the boundary $\gamma$ and thus lie in all Cauchy slices.
\end{proof}

Having established that the maximin surfaces of nested subregions are minimal on a common Cauchy slice, SSA and monogamy of mutual information follow quickly by adapting the original SSA argument of \cite{Headrick:2007km} and the monogamy argument of \cite{Hayden:2011ag}.  This discussion is identical to that of \cite{Wall:2012uf}, but we repeat it for completeness below.
\begin{cor}
Our restricted maximin construction satisfies strong subadditivity and negativity of tripartite mutual information (aka monogamy of mutual information).
\label{cor:ssa}
\end{cor}
\begin{proof}
Let $A,B$ and $C$ be non-overlapping regions in $\gamma$. Since $B\subset A\cup B\cup C$, we know that there exists a single surface $\Sigma\in\mathcal{C}_{\gamma}$
with $\Mm(B)$ and $\Mm(A\cup B\cup C)$ both minimal on $\Sigma$. As shown in figure \ref{fig:ssa}, the same surgery argument as in \cite{Headrick:2007km, Wall:2012uf} requires
\[
\A(\Mm(B))+\A(\Mm(A\cup B\cup C))\leq\A(\min(A\cup B,\Sigma))+\A(\min(B\cup C,\Sigma))
\]. But $\A(\min(A\cup B,\Sigma))<\A(\Mm(A\cup B))$
and $\A(\min(B\cup C,\Sigma))<\A(\Mm(B\cup C))$ by the definition of
$\Mm$.

\begin{figure}[h]
\centering
\includegraphics[width = 0.5\linewidth]{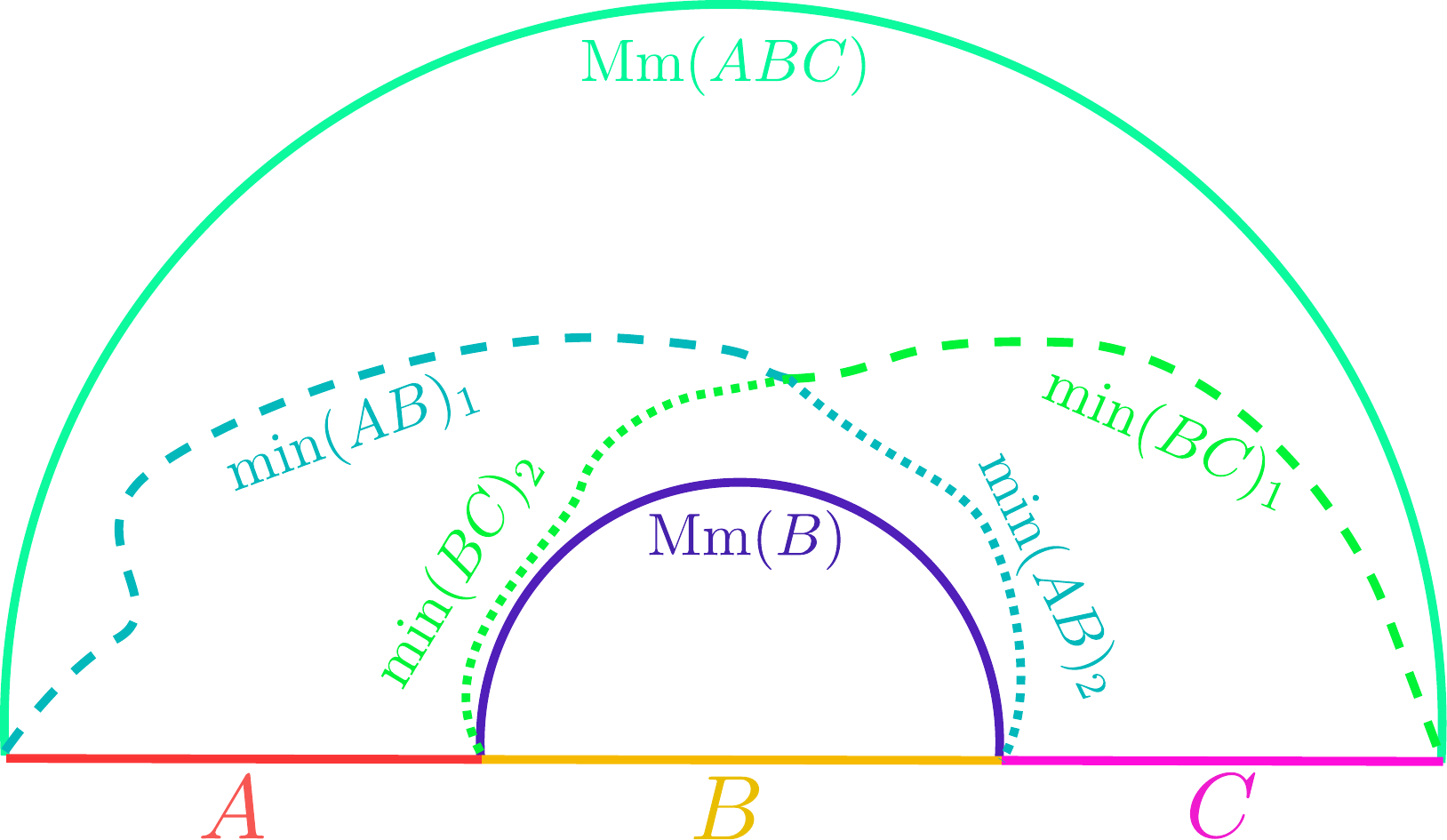}
\caption{The argument for strong subadditivity, with boundary subregions $A,B,$ and $C$. By Claim \ref{claim:ewn}, the surfaces $\Mm(ABC)$ and $\Mm(B)$ are minimal on a common Cauchy slice $\Sigma$. Note that we can also form a curve anchored to $\partial(ABC)$ by following dashed lines along min$(AB)_1$ and min$(BC)_1$.   A curve anchored to $\partial B$ can likewise be found by following the dotted lines along min$(AB)_2$ and min$(BC)_2$. However, minimality on $\Sigma$ means that $\A[\Mm(ABC)]\leq \A[\text{min}(AB)_1 \cup \text{min}(AB)_1]$ and $\A[\Mm(B)]\leq \A[\text{min}(AB)_2 \cup \text{min}(AB)_2]$.  Furthermore, $\A[\text{min}(AB)] \leq \A[\Mm(AB)]$ and  $\A[\text{min}(BC)] \leq \A[\Mm(BC)]$ by the maximization step of restricted maximin. Thus  $\A \Mm(ABC)  +\A \Mm(B) \leq \A(\text{min}(AB)_1) + \A(\text{min}(AB)_2) +\A(\text{min}(BC)_1)+\A(\text{min}(BC)_2) \leq \A(\Mm(AB)) + \A(\Mm(BC))$ andstrong subadditivity holds.  A similar argument holds in cases where regions $A,B,C$ do not meet at their boundaries.}
\label{fig:ssa}
\end{figure}

Monogamy of mutual information is handled similarly. We take $\Sigma$ such that $\Mm(A)$, $\Mm(B)$, $\Mm(C)$ and $\Mm(A\cup B\cup C)$ are all simultaneously minimized on $\Sigma$. then, the surgery argument implies that
\begin{align*}
 & \A(\Mm(A))+\A(\Mm(B))+\A(\Mm(C))+\A(\Mm(A\cup B\cup C))\\
\leq & \A(\min(A\cup B,\Sigma))+\A(\min(A\cup C,\Sigma))+\A(\min(B\cup C,\Sigma))
\end{align*}
from which the inequality follows. $\square$
\end{proof}

\section{Discussion}
\label{sec:disc}

Our work above studied possible definitions of holographic entropy in bulk spacetimes with a radial cutoff.  In particular, we focused on globally hyperbolic bulk regions that one may think of as domains of dependence for some achronal surface $\Sigma$ with boundary $\partial \Sigma = \gamma$.  We think of the surface $\gamma$ as defining the cutoff, even though it is codimension 2 in the full spacetime.  It may also be useful to think of $\gamma$ as an achronal slice of some codimension 1 timelike surface, whether the latter is a strict cutoff or just a partition of the bulk into two parts.  
While the areas of HRT surfaces anchored to $\gamma$ can generally violate  strong subadditivity (SSA), or even just subadditivity,  we argued that the areas of similarly anchored restricted maximin surfaces will satisfy both SSA and monogamy of mutual information.  Any other maximin-provable holographic entropy inequality \cite{Rota:2017ubr} will of course follow as well.  

The above qualifier `restricted' means that the maximization is only with respect to achronal surfaces $\Sigma$ with boundary $\partial \Sigma=\gamma$.  This restriction is important, as the example of section \ref{sec:counter-flat} shows that SSA can fail for the areas of unrestricted maximin surfaces in precisely the same manner as for HRT areas.  Appendix \ref{sec:appendix} establishes further properties of our restricted maximin surfaces not required for the main argument, but which may be useful in the future.  In particular,  up to sets of measure zero it shows that our restricted maximin surfaces either coincide with $\gamma$ or are spacelike separated from $\gamma$; i.e., null separations from $\gamma$ are rare.

As in the previous works \cite{Wall:2012uf,Marolf:2019bgj} we have neglected certain technical issues associated with proving that our restricted maximin surfaces can be chosen to be stable, and also with the fact that general Cauchy surfaces are not piece-wise smooth.  We hope that such points will be addressed in the near future, though we note that almost 30 years past before similar issues were fully resolved by \cite{Chrusciel:2000cu} in the context of the Hawking area theorem \cite{Hawking:1971tu}.

An interesting feature of our construction is that restricted maximin surfaces often have regions that coincide exactly with portions of $\gamma$.  This would not occur for an HRT surface except on portions of $\gamma$ that happen to be extremal.  However, it is clear that the same phenomenon {\it does} often occur for RT surfaces defined by minimizing areas over subsurfaces of some time-symmetric slice $\Sigma$. For example, in flat space it occurs whenever $\Sigma$ fails to be convex; see figure \ref{fig:concave}.  Our restricted maximin procedure is thus more similar to the RT prescription in the presence of a cutoff than is naive application of HRT.

\begin{figure}[h]
\centering
\includegraphics[width = 0.3\linewidth]{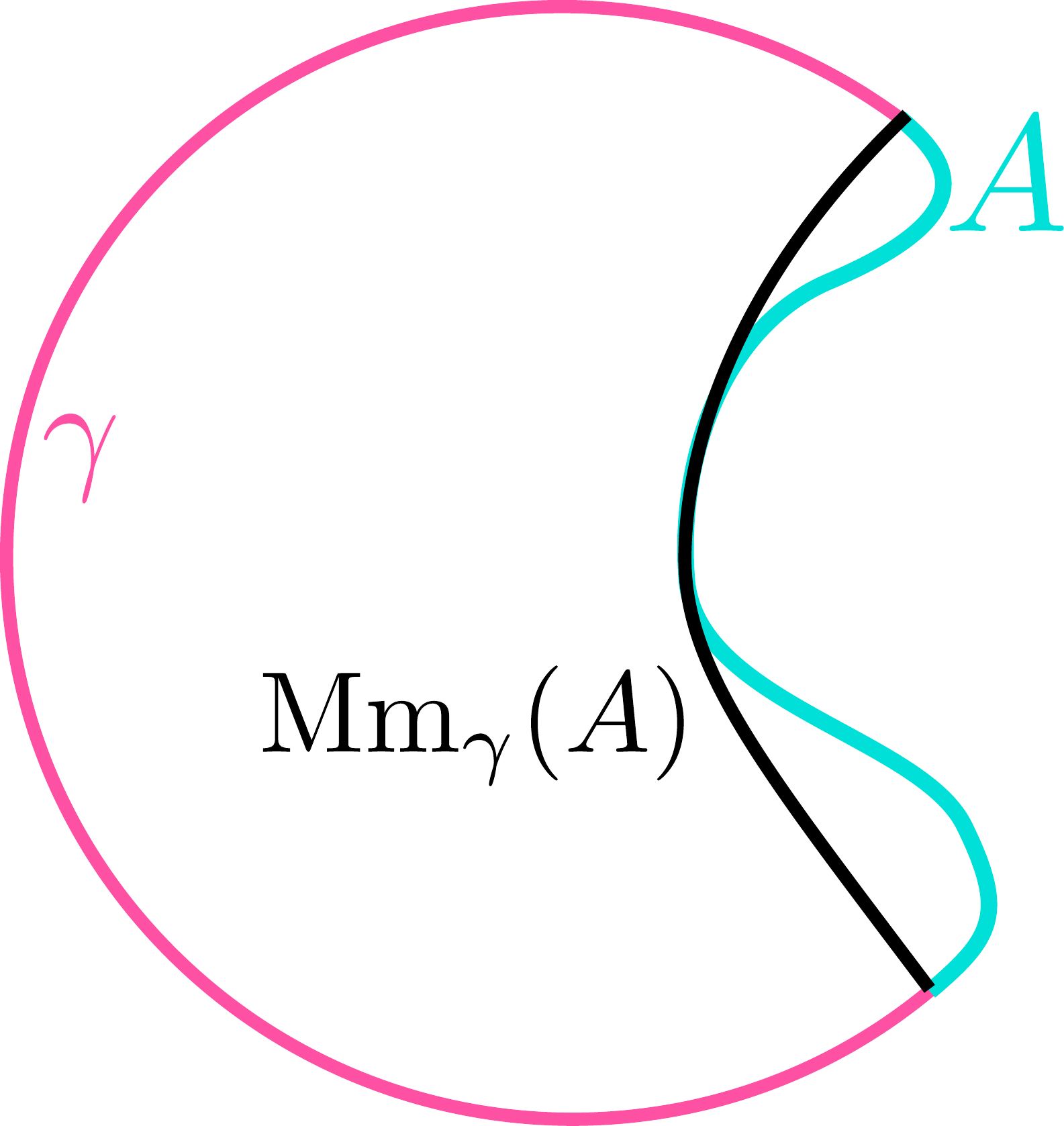}
\caption{A time-symmetric example where $\gamma$ fails to be convex, and consequently $\Mm(A)$ (and the RT surface associated with $A$) would coincide with $\gamma$ along some portion.}
\label{fig:concave}
\end{figure}

In retrospect, the reason that HRT surfaces behave so differently is simple to state.  It is well known that extrema of a function can occur either at stationary points, or at the edge of the allowed domain. In the absence of a cutoff and with natural boundary conditions, the latter can be ruled out.  But finite cutoffs provide an edge where extrema can arise, so an HRT-like prescription that involves solving differential equations for stationary points of the area functional can give very different results from applying a minimization procedure, or indeed from restricted maximin.

One would like to derive our prescription from an appropriate path integral in parallel with the cutoff-free arguments of \cite{Lewkowycz:2013nqa} and \cite{Dong:2016hjy}.  One class of obstacles to doing so are uncertainties regarding physics of the cutoff surface. If we momentarily ignore these issues we may proceed as in \cite{Donnelly:2018bef}; see also \cite{Donnelly:2019pie,Murdia:2019fax}. While  those references focused on convex spacetimes, in more general cases one would expect end-point effects similar to those described above to naturally arise in studying the path integral. For example, if we study replica geometries by preserving the original boundary conditions but introducing codimension 2 cosmic branes which source conical singularities, the path integral over cosmic brane locations will again often receive important contributions from configurations in which the brane partially coincides with the boundary of the spacetime.  This is particularly clear in the Euclidean context where the dominant contribution comes from the configuration with minimum action.  As above, this minimum may well arise on the boundary of the space of allowed cosmic brane solutions, and in particular where the cosmic brane itself partially coincides with the boundary of spacetime\footnote{We thank Xi Dong and Henry Maxfield for discussions of this point.}.

The fact that our restricted maximin procedure succeeds may be taken as evidence supporting the utility of cutoff holography.  However, it remains to better understand what further structures might be implied.  For example, the success of our procedure may also suggest that when a slice $\gamma$ of the cutoff surface is achronal with respect to the bulk causal structre, distinct points on $\gamma$ correspond to independent degrees of freedom for any bulk dual.  This is also natural from the bulk point of view.  But this viewpoint would in turn imply that the entropy of a boundary region $A$ should depend only on the appropriate notion of a domain of dependence of $A$, perhaps as determined by the causal structure of the bulk rather than the causal structure of the boundary.  And this statement is precisely what fails in the example of section \ref{sec:counter-flat}, where the area of the restricted maximin surface for $AB$ depends in detail on the choice of codimension 2 surface $\gamma.$  It would thus be interesting to investigate such issues further, perhaps by studying  $T\bar T$ deformations or their analogues in higher dimensions.

\section{Acknowledgments}
It is a pleasure to thank Xi Dong, Henry Maxifield, and especially Veronika Hubeny for useful discussions.  This material is based upon work supported by the Air Force Office of Scientific Research under award number FA9550-19-1-0360.  It  was also supported in part by funds from the University of California. In addition,
BG-W was supported in part by a University of California, Santa Barbara central fellowship.

\appendix 

\section{$\Mm(A)$ is mostly contained in the interior of $D(\gamma)$ and in $\gamma$ itself}
\label{sec:appendix}
This appendix contains some addition results concerning the possible intersections of restricted maximin surfaces with $\partial D(\gamma)$ that can be easily proven using the techniques of section \ref{sec:results}, but which are not required for our main results.  In particular, we saw in the first example from section \ref{sec:counter} that $\Mm(A)$ can intersect $\gamma$.  The proofs of the claims below establish that, up to possible sets of measure zero, $\Mm(A)$ is contained in the union of the interior of $D(\gamma)$ and $\gamma$ itself.  We focus below on excluding open sets of $\Mm(A)$ from $L^+$, but analogous arguments clearly also hold for $L^-$.

\begin{claim}
\label{theta_positive}
Suppose that $p\in L^{+}$ has $\theta_{+}\geq0$ (where $\theta_+$ is the expansion of the null generators of $L^+$). For any boundary region $A$, there does not exist a subset $V\subset \Mm(A), $ open in $\Mm(A)$, such that $p\in V $ and $V\subset L^+$.
\end{claim}

\begin{proof}
Suppose that such regions $A$ and $V$ exist.  There must exist a surface $\Sigma\in\mathcal{C}_{\gamma}$ such that $\Mm(A)=\min(A,\Sigma)$. Because $\Sigma$ is an achronal surface containing $V \subset L^+$ and $\gamma$, $\Sigma$ must contain the part of the null generator of $L^+$ containing $p$ that lies to the past of $p$.  Because $\theta_{+}\geq 0$ at $p$, the null generic condition requires $\theta_{+}>0$ for any point in $L^{+}$ to the past of $p$.  But since Claim \ref{claim:spacelike} forbids $\Mm(A)$ from lying along this generator of $L^+$, we see that deforming $\Mm(A)$ in the $-k$ direction at the point $p$ keeps the surface in $\Sigma$ and that this deformation reduces its area (since $\theta_+>0$).   This contradicts the fact that $\Mm(A)=\min(A,\Sigma)$, so $A,V$ cannot exist.
\end{proof}

\begin{claim}
\label{claim:negative_theta}
Suppose that $p\in L^{+}$ has $\theta_{+}<0$. Then $p$ cannot be in $\Mm(A)$ for any $A$.
\end{claim}
\begin{proof}

Let $\Sigma$ be a Cauchy surface on which $\Mm(A)$ is minimal, and consider the future-ingoing directed null congruence $N$ orthogonal to $\Mm(A)$.  Since it is nowhere to the future of $L^+$ and $\theta_{+}<0$, we must  have  $\theta_{k,N}<0$ by Corollary zero. But any $p\in \Mm(A)$ must have $\theta_{k,N}=0$ by Claim \ref{claim:ext}. So $p$ cannot lie in $L^+$.

\end{proof}


In fact, the argument for Claim \ref{claim:negative_theta} also applies at caustics or nonlocal intersections (i.e., at all points where $\theta_+$ is ill defined).  In such cases, deforming $\Sigma$ to the past again causes the area of any surface $\min(A,\Sigma_\epsilon)$ to exceed $\Mm(A)$, contradicting the maximization step of the maximin procedure.  We formalize this result in the following statement

\begin{claim}
\label{claim:caustic}
Suppose $\theta_+$ is ill defined or diverges at $p \in L^+$.  Thus $p \in L^{+}$ lies on a caustic or nonlocal self intersection of $L^+$ by Theorem 1 of \cite{Akers:2017nrr}.  Then for any boundary region $A$, there does not exist a subset $V\subset \Mm(A) $ open in $\Mm(A)$ such that $p\in V $ and $V\subset L^+$.
\end{claim}
\begin{proof}
Same as above.
\end{proof}

\bibliography{ssa}

\providecommand{\href}[2]{#2}\begingroup\raggedright\begin{thebibliography}{10}

\bibitem{McGough:2016lol}
L.~McGough, M.~Mezei, and H.~Verlinde, ``{Moving the CFT into the bulk with $
  T\overline{T} $},'' \href{http://dx.doi.org/10.1007/JHEP04(2018)010}{{\em
  JHEP} {\bfseries 04} (2018) 010},
\href{http://arxiv.org/abs/1611.03470}{{\ttfamily arXiv:1611.03470 [hep-th]}}.

\bibitem{Taylor:2018xcy}
M.~Taylor, ``{TT deformations in general dimensions},''
  \href{http://arxiv.org/abs/1805.10287}{{\ttfamily arXiv:1805.10287
  [hep-th]}}.

\bibitem{Hartman:2018tkw}
T.~Hartman, J.~Kruthoff, E.~Shaghoulian, and A.~Tajdini, ``{Holography at
  finite cutoff with a $T^2$ deformation},''
  \href{http://dx.doi.org/10.1007/JHEP03(2019)004}{{\em JHEP} {\bfseries 03}
  (2019) 004}, \href{http://arxiv.org/abs/1807.11401}{{\ttfamily
  arXiv:1807.11401 [hep-th]}}.

\bibitem{Guica:2019nzm}
M.~Guica and R.~Monten, ``{$T\bar T$ and the mirage of a bulk cutoff},''
  \href{http://arxiv.org/abs/1906.11251}{{\ttfamily arXiv:1906.11251
  [hep-th]}}.

\bibitem{Swingle:2009bg}
B.~Swingle, ``{Entanglement Renormalization and Holography},''
  \href{http://dx.doi.org/10.1103/PhysRevD.86.065007}{{\em Phys. Rev.}
  {\bfseries D86} (2012) 065007},
\href{http://arxiv.org/abs/0905.1317}{{\ttfamily arXiv:0905.1317
  [cond-mat.str-el]}}.

\bibitem{Qi:2013caa}
X.-L. Qi, ``{Exact holographic mapping and emergent space-time geometry},''
\href{http://arxiv.org/abs/1309.6282}{{\ttfamily arXiv:1309.6282 [hep-th]}}.

\bibitem{Evenbly2011}
G.~Evenbly and G.~Vidal, ``Tensor Network States and Geometry,''
  \href{http://dx.doi.org/10.1007/s10955-011-0237-4}{{\em Journal of
  Statistical Physics} {\bfseries 145} no.~4, (Nov, 2011) 891--918}.
  \url{https://doi.org/10.1007/s10955-011-0237-4}.

\bibitem{MolinaVilaplana:2011xt}
J.~Molina-Vilaplana and P.~Sodano, ``{Holographic View on Quantum Correlations
  and Mutual Information between Disjoint Blocks of a Quantum Critical
  System},'' \href{http://dx.doi.org/10.1007/JHEP10(2011)011}{{\em JHEP}
  {\bfseries 10} (2011) 011},
\href{http://arxiv.org/abs/1108.1277}{{\ttfamily arXiv:1108.1277 [quant-ph]}}.

\bibitem{Swingle:2012wq}
B.~Swingle, ``{Constructing holographic spacetimes using entanglement
  renormalization},''
\href{http://arxiv.org/abs/1209.3304}{{\ttfamily arXiv:1209.3304 [hep-th]}}.

\bibitem{Matsueda:2012xm}
H.~Matsueda, M.~Ishihara, and Y.~Hashizume, ``{Tensor network and a black
  hole},'' \href{http://dx.doi.org/10.1103/PhysRevD.87.066002}{{\em Phys. Rev.}
  {\bfseries D87} no.~6, (2013) 066002},
\href{http://arxiv.org/abs/1208.0206}{{\ttfamily arXiv:1208.0206 [hep-th]}}.

\bibitem{Pastawski:2015qua}
F.~Pastawski, B.~Yoshida, D.~Harlow, and J.~Preskill, ``{Holographic quantum
  error-correcting codes: Toy models for the bulk/boundary correspondence},''
  \href{http://dx.doi.org/10.1007/JHEP06(2015)149}{{\em JHEP} {\bfseries 06}
  (2015) 149},
\href{http://arxiv.org/abs/1503.06237}{{\ttfamily arXiv:1503.06237 [hep-th]}}.

\bibitem{Hayden:2016cfa}
P.~Hayden, S.~Nezami, X.-L. Qi, N.~Thomas, M.~Walter, and Z.~Yang,
  ``{Holographic duality from random tensor networks},''
  \href{http://dx.doi.org/10.1007/JHEP11(2016)009}{{\em JHEP} {\bfseries 11}
  (2016) 009},
\href{http://arxiv.org/abs/1601.01694}{{\ttfamily arXiv:1601.01694 [hep-th]}}.

\bibitem{Miyaji:2015yva}
M.~Miyaji and T.~Takayanagi, ``{Surface/State Correspondence as a Generalized
  Holography},'' \href{http://dx.doi.org/10.1093/ptep/ptv089}{{\em PTEP}
  {\bfseries 2015} no.~7, (2015) 073B03},
\href{http://arxiv.org/abs/1503.03542}{{\ttfamily arXiv:1503.03542 [hep-th]}}.

\bibitem{Takayanagi:2017knl}
T.~Takayanagi and K.~Umemoto, ``{Entanglement of purification through
  holographic duality},''
  \href{http://dx.doi.org/10.1038/s41567-018-0075-2}{{\em Nature Phys.}
  {\bfseries 14} no.~6, (2018) 573--577},
\href{http://arxiv.org/abs/1708.09393}{{\ttfamily arXiv:1708.09393 [hep-th]}}.

\bibitem{Nguyen:2017yqw}
P.~Nguyen, T.~Devakul, M.~G. Halbasch, M.~P. Zaletel, and B.~Swingle,
  ``{Entanglement of purification: from spin chains to holography},''
  \href{http://dx.doi.org/10.1007/JHEP01(2018)098}{{\em JHEP} {\bfseries 01}
  (2018) 098},
\href{http://arxiv.org/abs/1709.07424}{{\ttfamily arXiv:1709.07424 [hep-th]}}.

\bibitem{Bao:2018pvs}
N.~Bao, G.~Penington, J.~Sorce, and A.~C. Wall, ``{Beyond Toy Models:
  Distilling Tensor Networks in Full AdS/CFT},''
\href{http://arxiv.org/abs/1812.01171}{{\ttfamily arXiv:1812.01171 [hep-th]}}.

\bibitem{Bao:2019fpq}
N.~Bao, G.~Penington, J.~Sorce, and A.~C. Wall, ``{Holographic Tensor Networks
  in Full AdS/CFT},'' \href{http://arxiv.org/abs/1902.10157}{{\ttfamily
  arXiv:1902.10157 [hep-th]}}.

\bibitem{Krishnan:2019ygy}
C.~Krishnan, ``{Bulk Locality and Asymptotic Causal Diamonds},''
  \href{http://dx.doi.org/10.21468/SciPostPhys.7.4.057}{{\em SciPost Phys.}
  {\bfseries 7} no.~4, (2019) 057},
  \href{http://arxiv.org/abs/1902.06709}{{\ttfamily arXiv:1902.06709
  [hep-th]}}.

\bibitem{Krishnan:2020oun}
C.~Krishnan, V.~Patil, and J.~Pereira, ``{Page Curve and the Information
  Paradox in Flat Space},'' \href{http://arxiv.org/abs/2005.02993}{{\ttfamily
  arXiv:2005.02993 [hep-th]}}.

\bibitem{Ryu:2006ef}
S.~Ryu and T.~Takayanagi, ``{Aspects of Holographic Entanglement Entropy},''
  \href{http://dx.doi.org/10.1088/1126-6708/2006/08/045}{{\em JHEP} {\bfseries
  08} (2006) 045},
\href{http://arxiv.org/abs/hep-th/0605073}{{\ttfamily arXiv:hep-th/0605073
  [hep-th]}}.

\bibitem{Ryu:2006bv}
S.~Ryu and T.~Takayanagi, ``{Holographic derivation of entanglement entropy
  from AdS/CFT},'' \href{http://dx.doi.org/10.1103/PhysRevLett.96.181602}{{\em
  Phys. Rev. Lett.} {\bfseries 96} (2006) 181602},
\href{http://arxiv.org/abs/hep-th/0603001}{{\ttfamily arXiv:hep-th/0603001
  [hep-th]}}.

\bibitem{Donnelly:2018bef}
W.~Donnelly and V.~Shyam, ``{Entanglement entropy and $T \overline{T}$
  deformation},'' \href{http://dx.doi.org/10.1103/PhysRevLett.121.131602}{{\em
  Phys. Rev. Lett.} {\bfseries 121} no.~13, (2018) 131602},
  \href{http://arxiv.org/abs/1806.07444}{{\ttfamily arXiv:1806.07444
  [hep-th]}}.

\bibitem{Headrick:2007km}
M.~Headrick and T.~Takayanagi, ``{A Holographic proof of the strong
  subadditivity of entanglement entropy},''
  \href{http://dx.doi.org/10.1103/PhysRevD.76.106013}{{\em Phys. Rev.}
  {\bfseries D76} (2007) 106013},
\href{http://arxiv.org/abs/0704.3719}{{\ttfamily arXiv:0704.3719 [hep-th]}}.

\bibitem{Lewkowycz:2019xse}
A.~Lewkowycz, J.~Liu, E.~Silverstein, and G.~Torroba, ``{$T \bar T$ and EE,
  with implications for (A)dS subregion encodings},''
\href{http://arxiv.org/abs/1909.13808}{{\ttfamily arXiv:1909.13808 [hep-th]}}.

\bibitem{Geng:2019ruz}
H.~Geng, ``{Some Information Theoretic Aspects of De-Sitter Holography},''
  \href{http://dx.doi.org/10.1007/JHEP02(2020)005}{{\em JHEP} {\bfseries 02}
  (2020) 005}, \href{http://arxiv.org/abs/1911.02644}{{\ttfamily
  arXiv:1911.02644 [hep-th]}}.

\bibitem{Wall:2012uf}
A.~C. Wall, ``{Maximin Surfaces, and the Strong Subadditivity of the Covariant
  Holographic Entanglement Entropy},''
  \href{http://dx.doi.org/10.1088/0264-9381/31/22/225007}{{\em Class. Quant.
  Grav.} {\bfseries 31} no.~22, (2014) 225007},
\href{http://arxiv.org/abs/1211.3494}{{\ttfamily arXiv:1211.3494 [hep-th]}}.

\bibitem{Sanches:2016sxy}
F.~Sanches and S.~J. Weinberg, ``{Holographic entanglement entropy conjecture
  for general spacetimes},''
  \href{http://dx.doi.org/10.1103/PhysRevD.94.084034}{{\em Phys. Rev.}
  {\bfseries D94} no.~8, (2016) 084034},
\href{http://arxiv.org/abs/1603.05250}{{\ttfamily arXiv:1603.05250 [hep-th]}}.

\bibitem{Nomura:2018kji}
Y.~Nomura, P.~Rath, and N.~Salzetta, ``{Pulling the Boundary into the Bulk},''
  \href{http://dx.doi.org/10.1103/PhysRevD.98.026010}{{\em Phys. Rev. D}
  {\bfseries 98} no.~2, (2018) 026010},
  \href{http://arxiv.org/abs/1805.00523}{{\ttfamily arXiv:1805.00523
  [hep-th]}}.

\bibitem{Murdia:2020iac}
C.~Murdia, Y.~Nomura, and P.~Rath, ``{Coarse-Graining Holographic States: A
  Semiclassical Flow in General Spacetimes},''
  \href{http://arxiv.org/abs/2008.01755}{{\ttfamily arXiv:2008.01755
  [hep-th]}}.

\bibitem{Marolf:2019bgj}
D.~Marolf, A.~C. Wall, and Z.~Wang, ``{Restricted Maximin surfaces and HRT in
  generic black hole spacetimes},''
  \href{http://dx.doi.org/10.1007/JHEP05(2019)127}{{\em JHEP} {\bfseries 05}
  (2019) 127},
\href{http://arxiv.org/abs/1901.03879}{{\ttfamily arXiv:1901.03879 [hep-th]}}.

\bibitem{Hayden:2011ag}
P.~Hayden, M.~Headrick, and A.~Maloney, ``{Holographic Mutual Information is
  Monogamous},'' \href{http://dx.doi.org/10.1103/PhysRevD.87.046003}{{\em Phys.
  Rev.} {\bfseries D87} no.~4, (2013) 046003},
\href{http://arxiv.org/abs/1107.2940}{{\ttfamily arXiv:1107.2940 [hep-th]}}.

\bibitem{Penington:2019npb}
G.~Penington, ``{Entanglement Wedge Reconstruction and the Information
  Paradox},''
\href{http://arxiv.org/abs/1905.08255}{{\ttfamily arXiv:1905.08255 [hep-th]}}.

\bibitem{Almheiri:2019psf}
A.~Almheiri, N.~Engelhardt, D.~Marolf, and H.~Maxfield, ``{The entropy of bulk
  quantum fields and the entanglement wedge of an evaporating black hole},''
  \href{http://dx.doi.org/10.1007/JHEP12(2019)063}{{\em JHEP} {\bfseries 12}
  (2019) 063},
\href{http://arxiv.org/abs/1905.08762}{{\ttfamily arXiv:1905.08762 [hep-th]}}.

\bibitem{Almheiri:2019hni}
A.~Almheiri, R.~Mahajan, J.~Maldacena, and Y.~Zhao, ``{The Page curve of
  Hawking radiation from semiclassical geometry},''
\href{http://arxiv.org/abs/1908.10996}{{\ttfamily arXiv:1908.10996 [hep-th]}}.

\bibitem{Almheiri:2019yqk}
A.~Almheiri, R.~Mahajan, and J.~Maldacena, ``{Islands outside the horizon},''
\href{http://arxiv.org/abs/1910.11077}{{\ttfamily arXiv:1910.11077 [hep-th]}}.

\bibitem{Rozali:2019day}
M.~Rozali, J.~Sully, M.~Van~Raamsdonk, C.~Waddell, and D.~Wakeham,
  ``{Information radiation in BCFT models of black holes},''
\href{http://arxiv.org/abs/1910.12836}{{\ttfamily arXiv:1910.12836 [hep-th]}}.

\bibitem{Chen:2019uhq}
H.~Z. Chen, Z.~Fisher, J.~Hernandez, R.~C. Myers, and S.-M. Ruan,
  ``{Information Flow in Black Hole Evaporation},''
\href{http://arxiv.org/abs/1911.03402}{{\ttfamily arXiv:1911.03402 [hep-th]}}.

\bibitem{Bousso:2019ykv}
R.~Bousso and M.~Toma\v~sevi\'c, ``{Unitarity From a Smooth Horizon?},''
\href{http://arxiv.org/abs/1911.06305}{{\ttfamily arXiv:1911.06305 [hep-th]}}.

\bibitem{Almheiri:2019psy}
A.~Almheiri, R.~Mahajan, and J.~E. Santos, ``{Entanglement islands in higher
  dimensions},''
\href{http://arxiv.org/abs/1911.09666}{{\ttfamily arXiv:1911.09666 [hep-th]}}.

\bibitem{Almheiri:2019qdq}
A.~Almheiri, T.~Hartman, J.~Maldacena, E.~Shaghoulian, and A.~Tajdini,
  ``{Replica Wormholes and the Entropy of Hawking Radiation},''
\href{http://arxiv.org/abs/1911.12333}{{\ttfamily arXiv:1911.12333 [hep-th]}}.

\bibitem{Penington:2019kki}
G.~Penington, S.~H. Shenker, D.~Stanford, and Z.~Yang, ``{Replica wormholes and
  the black hole interior},''
\href{http://arxiv.org/abs/1911.11977}{{\ttfamily arXiv:1911.11977 [hep-th]}}.

\bibitem{Hubeny}
M.~Headrick and V.~Hubeny. To appear.

\bibitem{Casini:2006es}
H.~Casini and M.~Huerta, ``{A c-theorem for the entanglement entropy},''
  \href{http://dx.doi.org/10.1088/1751-8113/40/25/S57}{{\em J. Phys.}
  {\bfseries A40} (2007) 7031--7036},
\href{http://arxiv.org/abs/cond-mat/0610375}{{\ttfamily arXiv:cond-mat/0610375
  [cond-mat]}}.

\bibitem{Akers:2017nrr}
C.~Akers, R.~Bousso, I.~F. Halpern, and G.~N. Remmen, ``{Boundary of the future
  of a surface},'' \href{http://dx.doi.org/10.1103/PhysRevD.97.024018}{{\em
  Phys. Rev.} {\bfseries D97} no.~2, (2018) 024018},
\href{http://arxiv.org/abs/1711.06689}{{\ttfamily arXiv:1711.06689 [hep-th]}}.

\bibitem{Rota:2017ubr}
M.~Rota and S.~J. Weinberg, ``{New constraints for holographic entropy from
  maximin: A no-go theorem},''
  \href{http://dx.doi.org/10.1103/PhysRevD.97.086013}{{\em Phys. Rev. D}
  {\bfseries 97} no.~8, (2018) 086013},
  \href{http://arxiv.org/abs/1712.10004}{{\ttfamily arXiv:1712.10004
  [hep-th]}}.

\bibitem{Chrusciel:2000cu}
P.~T. Chrusciel, E.~Delay, G.~J. Galloway, and R.~Howard, ``{The Area
  theorem},'' \href{http://dx.doi.org/10.1007/PL00001029}{{\em Annales Henri
  Poincare} {\bfseries 2} (2001) 109--178},
  \href{http://arxiv.org/abs/gr-qc/0001003}{{\ttfamily arXiv:gr-qc/0001003}}.

\bibitem{Hawking:1971tu}
S.~Hawking, ``{Gravitational radiation from colliding black holes},''
  \href{http://dx.doi.org/10.1103/PhysRevLett.26.1344}{{\em Phys. Rev. Lett.}
  {\bfseries 26} (1971) 1344--1346}.

\bibitem{Lewkowycz:2013nqa}
A.~Lewkowycz and J.~Maldacena, ``{Generalized gravitational entropy},''
  \href{http://dx.doi.org/10.1007/JHEP08(2013)090}{{\em JHEP} {\bfseries 08}
  (2013) 090}, \href{http://arxiv.org/abs/1304.4926}{{\ttfamily arXiv:1304.4926
  [hep-th]}}.

\bibitem{Dong:2016hjy}
X.~Dong, A.~Lewkowycz, and M.~Rangamani, ``{Deriving covariant holographic
  entanglement},'' \href{http://dx.doi.org/10.1007/JHEP11(2016)028}{{\em JHEP}
  {\bfseries 11} (2016) 028}, \href{http://arxiv.org/abs/1607.07506}{{\ttfamily
  arXiv:1607.07506 [hep-th]}}.

\bibitem{Donnelly:2019pie}
W.~Donnelly, E.~LePage, Y.-Y. Li, A.~Pereira, and V.~Shyam, ``{Quantum
  corrections to finite radius holography and holographic entanglement
  entropy},'' \href{http://dx.doi.org/10.1007/JHEP05(2020)006}{{\em JHEP}
  {\bfseries 05} (2020) 006}, \href{http://arxiv.org/abs/1909.11402}{{\ttfamily
  arXiv:1909.11402 [hep-th]}}.

\bibitem{Murdia:2019fax}
C.~Murdia, Y.~Nomura, P.~Rath, and N.~Salzetta, ``{Comments on holographic
  entanglement entropy in $TT$ deformed conformal field theories},''
  \href{http://dx.doi.org/10.1103/PhysRevD.100.026011}{{\em Phys. Rev. D}
  {\bfseries 100} no.~2, (2019) 026011},
  \href{http://arxiv.org/abs/1904.04408}{{\ttfamily arXiv:1904.04408
  [hep-th]}}.

\end{thebibliography}\endgroup
\bibliographystyle{utcaps}

\end{document}